\documentclass[a4paper,twocolumn,11pt,accepted=2024-11-09]{quantumarticle}
\pdfoutput=1
\usepackage[utf8]{inputenc}
\usepackage[english]{babel}
\usepackage[T1]{fontenc}
\usepackage{amsmath}
\usepackage{hyperref}

\usepackage{lipsum}
\usepackage[numbers,sort&compress]{natbib}
\usepackage{graphicx}
\usepackage{bm}
\usepackage{amsmath}
\usepackage{amssymb}
\usepackage{enumitem}
\usepackage{appendix}
\usepackage{comment}
\usepackage{relsize}
\usepackage[bottom]{footmisc}
\usepackage[dvipsnames]{xcolor}
\definecolor{darkblue}{rgb}{0,0,.65}
\definecolor{darkgreen}{rgb}{0.28,0.41,0.19}
\usepackage[dvipsnames]{xcolor}

\usepackage[all]{hypcap}
\usepackage{braket}
\usepackage{amsthm,bbm}

\usepackage{lipsum}
\usepackage{tcolorbox}

\def\RR{\mathbb{R}}

\newcommand{\bd}[1]{\boldsymbol{#1}}
\newcommand{\wb}{\bm{w}}
\newcommand{\wtb}{\bd{\tilde{w}}}
\newcommand{\Ew}{E_{\wb}}
\newcommand{\Eb}{\mathbf{E}}
\newcommand{\Ebt}{\mathbf{\tilde{E}}}
\newcommand{\rw}{\rho_{\wb}}
\newcommand{\rtw}{\tilde{\rho}_{\wb}}

\newcommand{\Dw}{\mathcal{D}_{\wb} }
\newcommand{\DEwH}{\Delta E_{\wb}}
\newcommand{\Xb}{\mathbf{X}}

\newcommand{\deltarhow}{\Delta \rho_{\wb}}
\newcommand{\deltaPsik}{\Delta \Psi_k}
\newcommand{\deltaPsi}{\Delta \Psi}
\newcommand{\deltaEk}{\Delta E_k}
\newcommand{\deltaE}{\Delta E}
\newcommand{\mub}{\bm{\mu}}

\newtheorem{thrm}{Theorem}
\newtheorem{coro}{Corollary}
\newtheorem{lemm}{Lemma}

\begin{document}

\title{Ground and Excited States from Ensemble Variational Principles}

\author{Lexin Ding}
\affiliation{Faculty of Physics, Arnold Sommerfeld Centre for Theoretical Physics (ASC),\\Ludwig-Maximilians-Universit{\"a}t M{\"u}nchen, Theresienstr.~37, 80333 M{\"u}nchen, Germany}
\affiliation{Munich Center for Quantum Science and Technology (MCQST), Schellingstrasse 4, 80799 M{\"u}nchen, Germany}

\author{Cheng-Lin Hong}
\affiliation{Faculty of Physics, Arnold Sommerfeld Centre for Theoretical Physics (ASC),\\Ludwig-Maximilians-Universit{\"a}t M{\"u}nchen, Theresienstr.~37, 80333 M{\"u}nchen, Germany}
\affiliation{Munich Center for Quantum Science and Technology (MCQST), Schellingstrasse 4, 80799 M{\"u}nchen, Germany}

\author{Christian Schilling}
\email{c.schilling@lmu.de}
\affiliation{Faculty of Physics, Arnold Sommerfeld Centre for Theoretical Physics (ASC),\\Ludwig-Maximilians-Universit{\"a}t M{\"u}nchen, Theresienstr.~37, 80333 M{\"u}nchen, Germany}
\affiliation{Munich Center for Quantum Science and Technology (MCQST), Schellingstrasse 4, 80799 M{\"u}nchen, Germany}

\maketitle

\begin{abstract}
The extension of the Rayleigh-Ritz variational principle to ensemble states $\rho_{\wb}\equiv\sum_k w_k |\Psi_k\rangle \langle\Psi_k|$ with fixed weights $w_k$ lies ultimately at the heart of several recent methodological developments for targeting excitation energies by variational means. Prominent examples are density and density matrix functional theory, Monte Carlo sampling, state-average complete active space self-consistent field methods and variational quantum eigensolvers. In order to provide a sound basis for all these methods and to improve  their current implementations, we prove the validity of the underlying critical hypothesis: Whenever the ensemble energy is well-converged, the same holds true for the ensemble state $\rw$ as well as the individual eigenstates $\ket{\Psi_k}$ and eigenenergies $E_k$. To be more specific, we derive linear bounds $d_-\Delta\Ew \leq  \Delta Q \leq d_+ \Delta\Ew$ on the errors $\Delta Q $ of these sought-after quantities. A subsequent analytical analysis and numerical illustration proves the tightness of our universal inequalities. Our results and particularly the explicit form  of $d_{\pm}\equiv d_{\pm}^{(Q)}(\wb,\Eb)$ provide valuable insights into the optimal choice of the auxiliary weights $w_k$ in practical applications.
\end{abstract}

\section{Introduction}
Quantum excitations play a crucial role in the understanding of quantum systems in general, and they lie at the heart of numerous applications in
physics, chemistry, materials science and biology~\cite{Green2014,Johnson2015,Yam2023,Cheng2009,Cerullo2395}.
Unlike ground states, which can nowadays be numerically calculated with high accuracy and ever-increasing efficiency~\cite{PhysRevX.5.041041,PhysRevX.10.031058},
the study of excited-state properties is still a long-standing computational challenge~\cite{RevModPhys.74.601,SERRANOANDRES200599,doi:10.1021/acs.chemrev.8b00244,Dash2019,PhysRevLett.126.150601,Reiher2021,Westermayr,Loos2018}.
The complexity of the latter lies in the multi-configurational character of the corresponding wave functions~\cite{Gonzalez2012},
as well as in the lack of a variational principle for targeting individual excited states exclusively.

As far as the computation of ground states is concerned, the effectiveness of numerous numerical methods~\cite{PhysRev.136.B864,RevModPhys.77.259,Carleo2017} rests on the underlying Rayleigh-Ritz (RR) variational principle: For any Hamiltonian
\begin{equation}\label{Ham}
\hat{H} \equiv \sum_{k=0}^{D-1} E_k |\Psi_k\rangle\langle\Psi_k|
\end{equation}
on a $D$-dimensional Hilbert space $\mathcal{H}$ 
the ground state energy $E_0$ can be obtained by variational means according to
\begin{equation}\label{RR}
    E_0 \leq \langle \Psi|\hat{H}|\Psi\rangle, \quad \forall |\Psi\rangle \in \mathcal{H}.
\end{equation}
The extremality within the energy spectrum of the sought-after minimal eigenvalue $E_0$ of $\hat{H}$ is absolutely crucial for the effectiveness and predictive power of variational ground state approaches.
It namely implies (see, e.g., Ref.~~\cite{Benavides2017gap} and Sec.~\ref{sec:gok}) that the accuracy of a trial state $\ket{\Psi}$ increases as the variational energy $\bra{\Psi}H \ket{\Psi}$ approaches the ground state energy $E_0$, and equality with the exact ground state is reached if and only if the variational energy equals $E_0$. 

In order to target the low-lying excitation spectrum by variational means, Gross, Oliveira and Kohn~\cite{GOK1988} extended the RR variational principle to ensemble states $\rtw \equiv \sum_{k=0}^{K-1} w_k |\tilde{\Psi}_k\rangle\langle\tilde{\Psi}_k|$ with fixed spectrum $\wb$:
\begin{equation}
\begin{split}
    \Ew \equiv  \sum_{k=0}^{K-1} w_k E_k &\leq \mathrm{Tr}[\rtw \hat{H}]
    \\
    &= \sum_{k=0}^{K-1} w_k \langle \tilde{\Psi}_k|\hat{H}  |\tilde{\Psi}_k\rangle.
\end{split}
\end{equation}
The first $K$ eigenenergies $E_k$ and eigenstates $\ket{\Psi_k}$ of $\hat{H}$ are then obtained from the spectral decomposition of the resulting minimizer $\rw \equiv \sum_{k=0}^{K-1} w_k |\Psi_k\rangle\langle\Psi_k|$.

\begin{figure}[t]
    \centering
    \includegraphics[width=\linewidth]{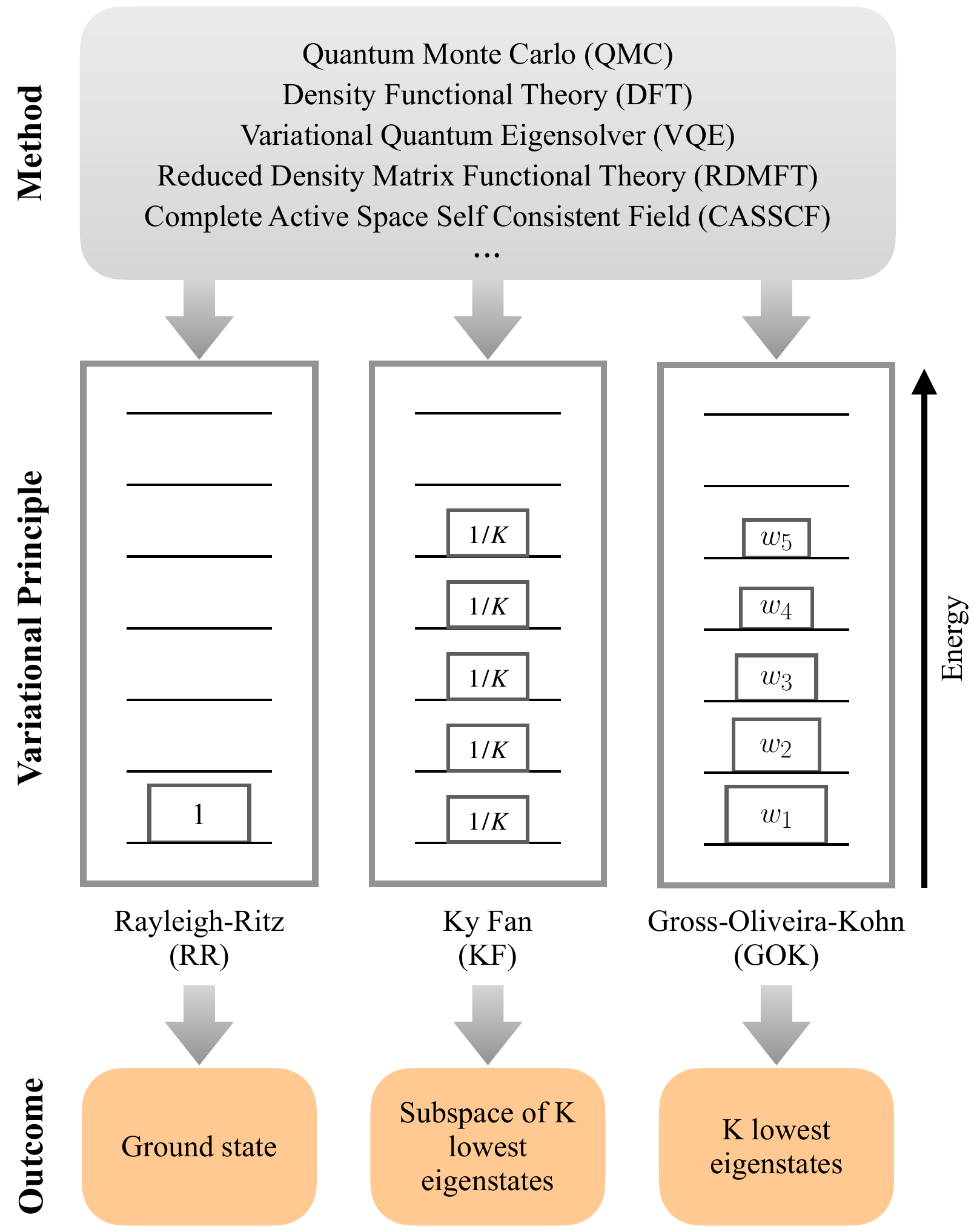}
    \caption{Illustration of the Rayleigh-Ritz (RR), Ky Fan (KF), and Gross-Oliveira-Kohn (GOK) variational principles and their scope in the context of recent methodological developments of ground and excited state approaches.}
    \label{fig:scope}
\end{figure}

This Gross-Oliveira-Kohn (GOK) variational principle naturally extends the scope of ground state methods to excited states (see the overview in Figure \ref{fig:scope}), and has recently become a key ingredient for the conceptual and numerical developments in quantum sciences due to a number of reasons:
First, since the common approach in density functional theory (DFT)~\cite{HK} to target excited states through a time-dependent formulation~\cite{RG84,Gross96ExcDFT} suffers from conceptual and practical limitations, ensemble/GOK-DFT~\cite{PhysRevA.37.2809,Carsten2017,PhysRevB.95.035120,Fromager18,loos2020weight,Fromager20, Fromager21,Gross21,gould2021ensemble,PhysRevA.104.052806,lu2022multistate,Loos23JPCLett,PhysRevB.109.235113} has emerged as a promising alternative. Second, due to the difficulties of DFT in dealing with strong correlation effects, a one-particle reduced density matrix functional theory (RDMFT)~\cite{Gilbert,LE79,V80,LCS23} based on the GOK variational principle for targeting selected eigenstates has been proposed in 2021 ($\wb$-RDMFT)~\cite{PhysRevLett.127.023001,liebert2021foundation,JLthesis,LS23SciP,LS23NJP}. Third, wave function-based approaches that calculate the lowest $K$ eigenstates consecutively suffer from an uncontrollable accumulation of errors. Utilizing instead the GOK variational principle with flexible weights $\wb$ provides a more direct approach to the specific eigenstates of interest, which is successfully exploited, e.g., in quantum Monte Carlo techniques~\cite{schautz2004optimized,schautz2004excitations,Filippi}, some traditional quantum chemical methods~\cite{mcweeny1974scf,sachs1975mcscf,deskevich2004dynamically,PhysRevLett.129.066401,MatsikaRev21} and quantum computing~\cite{ssvqe,Higgott2019variationalquantum,Cerezo2022,Nakanishi,Yalouz_2021,Yalouz22,PhysRevA.107.052423,hong2023quantum,benavides2023quantum}.

However, despite its broad applicability in recent years, two crucial aspects regarding the GOK variational principle have not been properly formulated nor addressed yet. (i) There are two seemingly misaligned objectives which first need to be reconciled: It is the ensemble energy that is minimized, yet the interest lies in the individual eigenstates and their energies. Although reaching the exact ensemble energy $\Ew$ implies that the eigenstates of the ensemble are exact as well~\cite{GOK1988}, it is from a mathematical point of view not clear at all that the GOK variational principle maintains this predictive power for the practically relevant case of \emph{imperfect} convergence. Indeed, particularly for quasidegenerate weights $w_k$, one could easily imagine an incorrect ensemble $\rtw= \sum_{k\geq 0} w_k |\tilde{\Psi}_k\rangle\langle\tilde{\Psi}_k| \neq\rw$ which, however, (almost) attains the exact energy $\Ew$ due to a cancellation of individual energy errors $\tilde{E}_k-E_k$.
(ii) There is an under-explored degree of freedom in the choice of $\wb$. The validity of the GOK variational principle is independent of the specific choices of weights. Yet clearly, a variational algorithm can be biased towards a specific eigenstate, if it is assigned a larger weight. One extreme example would be the reduction of the GOK variational principle to the original RR version, or the one of equal weights, when virtually no information regarding the individual eigenstates could be gained.
It is the ultimate goal of our work to address both of the above issues comprehensively and settle them conclusively.

The paper is structured as follows. First, we discuss the GOK variational principle and point out related caveats in Sec.~\ref{sec:gok}, formulate concisely our scientific problem in Sec.~\ref{sec:problem} and introduce in Sec.~\ref{sec:poly} two elegant concepts from convex geometry required for its solution. In Sec.~\ref{sec:bounds}, we present our main results, namely universal bounds, $d_-\Delta\Ew \leq  \Delta Q \leq d_+\Delta\Ew$, on the error $\Delta Q$ of the variational ensemble state, the eigenstates and eigenenergies.
In Sec.~\ref{sec:saturation}, we then confirm by analytically and numerically means the tightness of these bounds and identify optimal choices of the weights for practical applications based on the explicit form of $d_\pm\equiv d_\pm^{(Q)}(\wb,\Eb)$. Last but not least, in Sec.~\ref{sec:VQE}, as a proof of concept, we illustrate our results in the context of variational quantum eigensolvers.

\section{Notations, Concepts and Scientific Problem} \label{sec:concepts}

\subsection{GOK variational principle} \label{sec:gok}
In this section we provide a brief recap of the GOK variational principle, identify some crucial caveats in its practical applications and explain what is required from a variational principle in order to have a sound basis.

Attempts to extend the RR ground state variational principle date back to the late 40s, when Kohn and Ky Fan (the latter provided a full proof) established that the energy of an equiensemble state of rank $K$ is bounded from above by the sum of the $K$ lowest energy eigenvalues~\cite{kohn1947excited,fan1949theorem}:
\begin{thrm}[Ky Fan (1949)]\label{thrm:kyfan}
    For any Hamiltonian \eqref{Ham} and any integer $K \leq D$ the inequality
\begin{equation}
    \sum_{k=0}^{K-1} E_k^\uparrow \leq  \sum_{k=0}^{K-1} \langle \tilde{\Psi}_k |\hat{H}|\tilde{\Psi}_k\rangle
\end{equation}
holds for any set of $K$ orthonormal states $\{|\tilde{\Psi}_k\rangle\}_{k=0}^{K-1}$. Moreover, as long as $E_K > E_{K-1}$, equality is attained if and only if the span of the latter  coincides with the span of the first $K$ eigenstates $\{|\Psi_k\rangle\}_{k=0}^{K-1}$.
\end{thrm}
\noindent Here and in the following, an arrow $\uparrow (\downarrow)$ indicates that the entries of a vector are arranged increasingly (decreasingly).

It was almost 40 years later when  Gross, Oliveira and Kohn~\cite{GOK1988}
further extended Theorem \ref{thrm:kyfan} to ensembles with arbitrary weights:
\begin{thrm}[Gross, Oliveira, Kohn (1988)]\label{thrm:gok}
    For any Hamiltonian \eqref{Ham} and any probability distribution ${\wb}\in\RR^D$ ($w_i\geq 0$
    and $\sum_{i=0}^{D-1}w_i = 1$), the inequality
\begin{equation}\label{GOKineq}
 {E}_{\wb} \equiv  \sum_{k=0}^{D-1} w_k^\downarrow E_k^\uparrow   \leq   \mathrm{Tr}[\tilde{\rho}_{\wb}\hat{H}] = \sum_{k=0}^{D-1} w_k \bra{\tilde{\Psi}_k} \hat{H} \ket{\tilde{\Psi}_k}.
\end{equation}
holds for any ensemble state $\tilde\rho_{{\wb}} \equiv \sum_{k=0}^{D-1} w_k |\tilde{\Psi}_k\rangle\langle\tilde{\Psi}_k|$ with spectrum $\wb$.
Moreover, as long as the energy spectrum $\Eb$ of $\hat{H}$ is non-degenerate, equality in \eqref{GOKineq} is attained if and only if $\rtw$ and $\rw \equiv \sum_{k=0}^{D-1} w_k^\downarrow |\Psi_k\rangle\langle \Psi_k|$, and thus their spectral decompositions, coincide.
\end{thrm}

A couple of comments are in order here concerning this GOK variational principle.
First, the weighted energy average is nothing else than the energy expectation value of the exact $\wb$-ensemble state $\rw$, i.e., $ {E}_{\wb} \equiv  \sum_{k=0}^{D-1} w_k^\downarrow E_k^\uparrow   =   \mathrm{Tr}[\rho_{\wb}\hat{H}] = \sum_{k=0}^{D-1} w_k \bra{\Psi_k} \hat{H} \ket{\Psi_k}$ and thus \eqref{GOKineq} could be recast as $\mathrm{Tr}[(\tilde{\rho}_{\wb}-\rho_{\wb})\hat{H}] \geq 0$ for all $\wb$-ensemble states $\rtw$. Second, the predictive power of the GOK variational principle 
depends on how degenerate the energy spectrum $\Eb$ is. To explain this, let us consider the (practically less interesting) case of a fully degenerate spectrum, i.e., $\hat{H}= E_0 \hat{\openone}$. Any $\wb$-ensemble state would then saturate inequality \eqref{GOKineq}. Accordingly, the information that a given $\rtw$ saturates \eqref{GOKineq} would not tell us anything about this ensemble state. On the contrary, for a non-degenerate spectrum $E_0 < E_1 < \ldots < E_{D-1}$, the saturation of \eqref{GOKineq} by some $\rtw$ would imply that its eigenbasis $\{\ket{\tilde{\Psi}_k}\}_{k=0}^{D-1}$ coincides (up to irrelevant phase factors) with the eigenbasis $\{\ket{\Psi_k}\}_{k=0}^{D-1}$ of $\hat{H}$ and $\rw$, respectively. Given the goal of our paper, namely to elaborate on and confirm the predictive power of the GOK variational principle, this comment would suggest at first sight to study all scenarios of different energy multiplicities. Yet, there are good reasons to restrict our work to non-degenerate energy spectra only: In practice, degeneracies typically emerge from symmetries of the Hamiltonian, which translates into a block-diagonal structure, $\hat{H} \equiv \oplus_{Q}\hat{H}_Q$, with respect to the different symmetry sectors. One could then apply the entire reasoning of our work to each non-degenerate block/symmetry sector separately. Moreover, as our main results in Sec.~\ref{sec:bounds} will confirm, the case of degeneracies is often either already covered as a suitable limit of the non-degenerate case, or alternatively it just does not make sense anymore to elaborate on the predictive power of GOK (as the example $\hat{H}= E_0 \hat{\openone}$ above illustrated). Third, in applications of the GOK variational principle to realistic quantum-many body systems one is typically interested only in the lowest few ($K$) eigenstates in the exponentially large Hilbert space. Hence, one then considers only $\wb$-vectors with $K$ non-vanishing weights, $w_0 \!>\! w_1 \!>\! \cdots \!>\! w_{K-1} \!>\! w_K \!=\! \cdots \!=\!w_{D-1} \!=\!0$. Of course, this in turn implies that the predictive power of the GOK variational principle then restricts to the lowest $K$ eigenenergies and eigenstates, and that possible degeneracies of the higher energy levels $E_k$, $k \geq K$, will not affect our derivation and results anymore.
Fourth, while a brute-force proof of Theorem \ref{thrm:gok} was given in Ref.~~\cite{GOK1988}, we will provide a particularly compact and simple one in Sec.~\ref{sec:poly}, after having introduced some effective tools from convex analysis.


According to Theorem \ref{thrm:gok}, the ensemble energy $\tilde{E}_{\wb}$ of any $\rtw \equiv \sum_{k=0}^{D-1} w_k^\downarrow |\tilde{\Psi}_k\rangle\langle\tilde{\Psi}_k|$ is an upper bound to the true minimum $E_{\wb}$. This, however, does not establish yet a sound basis for a variational principle that can be used for calculating individual eigenstates and their energies. Indeed, when the minimum $\Ew$ is not exactly attained, Theorem \ref{thrm:gok} does not make any predictions about the meaning or quality of the individual states $\ket{\tilde{\Psi}_k}$. That is quite in contrast to ground state calculations through the RR variational principle: It is well understood that for any pure trial state $\ket{\Psi}\!\bra{\Psi}$, its Hilbert-Schmidt distance to the exact ground state $\ket{\Psi_0}\!\bra{\Psi_0}$ ($\|\hat{A}\|_\mathrm{HS}\equiv\sqrt{\mathrm{Tr}[\hat{A}^\dagger \hat{A}]}$),
\begin{equation}
\Delta  \Psi_0 \equiv \frac{1}{2} \big\lVert|\Psi\rangle\langle\Psi|-|\Psi_0\rangle\langle\Psi_0|\big\rVert_\mathrm{HS}^2 = 1-|\langle \Psi | \Psi_0\rangle|^2,
\end{equation}
is linearly bounded from below and above by the variational error $\Delta E_0 \equiv \langle \Psi|\hat{H}|\Psi\rangle-E_0 \geq 0$ according to (see, e.g., Ref.~~\cite{Benavides2017gap})
\begin{equation}
q_- \,\Delta E_0 \leq  \Delta  \Psi_0\leq q_+\, \Delta E_0, \label{eqn:RR_state}
\end{equation}
where $q_-\equiv1/(E_{D-1}\!-\!E_0)$ and $q_+\equiv1/(E_1\!-\!E_0)$.
This means nothing else than that the variational state is exact if and only if the energy is fully converged and otherwise the errors $\Delta E_0, \Delta  \Psi_0$ are related through linear bounds. It is exactly this \emph{predictive power} that makes the RR variational principle an effective and meaningful variational principle. It is therefore the ultimate goal of our work to elaborate on and eventually confirm such predictive power also for the GOK variational principle. This means to understand in quantitative terms whether, when and how the ensemble error $\Delta \Ew$ bounds the errors of (i) the ensemble state $\rho_{{\wb}}$, (ii) the individual eigenstates $ |\Psi_k\rangle$, and (iii) the individual eigenenergies $E_k$. This central task of our work can become quite delicate --- in contrast to the context of the RR variational principle --- as the following basic example illustrates:
\begin{tcolorbox}[colback=Green!20, colframe=Green, title= An instructive example]
Consider a three-level system with Hamiltonian $\hat{H} = -|0\rangle\langle0| + |2\rangle\langle 2|$ and choose the ensemble weights ${\wb} = (1/2,1/2,0)$. In this case both density operators $\rho_{{\wb}} = \frac{1}{2}|0\rangle\langle 0| + \frac{1}{2}|1\rangle\langle1|$ and $\rtw = \frac{1}{2}|+\rangle\langle +| + \frac{1}{2}|-\rangle\langle-|$, where $\ket{\pm} \equiv 1/\sqrt{2}(\ket{0}\pm \ket{1})$, are minimizers of the ensemble energy (both states are actually the same). However, the states $\ket{\pm}$ are quite different from the first two eigenstates $\ket{0}, \ket{1}$ of $\hat{H}$ and accordingly have large errors in their energies.
\end{tcolorbox}
While the negative consequences of a degenerate weight vector $\wb$ were actually rather obvious, this basic example clearly reveals the depth of the scientific problem at hand: the anticipated error bounds must depend in a non-trivial way on (the gaps between) the weights $w_k$, and in  particular whenever two weights become identical some bounds must become meaningless.

\subsection{The scientific problem} \label{sec:problem}
In this section, we formulate our scientific problem in concise terms. For this, we first assume a fixed non-degenerate Hamiltonian $\hat{H}$ on a $D$-dimensional complex Hilbert space $\mathcal{H}$ according to \eqref{Ham} and denote its vector of increasingly ordered energy eigenvalues by $\Eb\equiv \Eb^\uparrow$. For each choice of a spectral weight vector $\wb\equiv \wb^\downarrow$, we then define the corresponding manifold of density operators with spectrum $\wb$ to which GOK's variational principle, Theorem \ref{thrm:gok}, refers to:
\begin{equation}\label{Dw}
    \mathcal{D}_{\wb} \equiv \left\{ \tilde{\rho}_{{\wb}} = \left.\sum_{k=0}^{D-1} w_k |\tilde{\Psi}_k\rangle\langle\tilde{\Psi}_k| \, \right|\, \langle\tilde{\Psi}_k|\tilde{\Psi}_l\rangle = \delta_{kl} \right\}.
\end{equation}
It is worth noticing that the manifold $\mathcal{D}_{\wb}$ is identical to the set of all unitary conjugations of the exact ensemble state $\rw \equiv \sum_{k=0}^{D-1} w_k |\Psi_k\rangle\langle\Psi_k|$, i.e.,
\begin{equation}
\Dw = \big\{\hat{U} \rho_{\wb} \hat{U}^\dagger \,|\, \hat{U} \mbox{ a unitary on }\mathcal{H}\big\}
\end{equation}
and accordingly every element $\tilde{\rho}_{\wb}$ can be expressed  as
\begin{equation}\label{rhoU}
    \rtw(U) = \hat{U} \rho_{\wb} \hat{U}^\dagger
\end{equation}
for some unitary $\hat{U}$.
To each quantum state $\tilde{\rho}_{\wb}$ in $\mathcal{D}_{\wb}$ we can associate the corresponding vector $\Ebt \equiv (\tilde{E}_k)_{k=0}^{D-1}$ of energy expectation values $\tilde{E}_k = \bra{\tilde{\Psi}_k} \hat{H} \ket{\tilde{\Psi}_k}$ with respect to $\rtw$'s eigenbasis $\{\ket{\tilde{\Psi}_k}\}_{k=0}^{D-1}$.

Elaborating on the predictive power of the GOK variational principle means nothing else than to predict the errors of physical properties of variational states $\rtw \in \Dw$ as function of their ensemble energy error
\begin{equation}
\DEwH(\rtw) \equiv  \mathrm{Tr}[(\tilde{\rho}_{{\wb}}-\rho_{{\wb}})\hat{H}]. \label{rhoerror}
\end{equation}
Due to its  relevance for numerical calculations in practice, a particular emphasis will lie on the regime where the ensemble energy error becomes small. With these definitions at hand, we can now formulate in concise terms the scientific problem that we are going to solve:
\begin{tcolorbox}[colback=blue!20, colframe=blue, title= \hspace{-3.5mm} \mbox{Scientific Problem: Predictive Power of GOK}]
Consider the GOK variational principle (Theorem \ref{thrm:gok}) for a setting $(\mathcal{H},\hat{H},\wb)$. Determine for any relevant quantity $Q \equiv Q(\rtw)$ the range
\begin{equation}
d_-^{(Q)}(\delta,\wb,\hat{H}) \leq \Delta Q(\rtw)\leq  d_+^{(Q)}(\delta,\wb,\hat{H}), \label{thmineq}
\end{equation}
of its possible errors  $\Delta Q(\rtw) \equiv Q(\rtw)- Q(\rw)$ provided $\rtw$ has a (sufficiently small) energy error $\DEwH(\rtw)=\delta>0$. The relevant quantities are \\
(i) the ensemble state $\rw$\\
(ii) all individual eigenstates $\ket{\Psi_k}$ \\
(iii) all individual eigenenergies $E_k$
\end{tcolorbox}

\noindent A couple of comments are in order. First, the sought-after optimal bounds are given by
\begin{eqnarray}\label{dbounds}
  d_-^{(Q)}(\delta,\wb,\hat{H}) &=& \min_{\scriptsize\begin{array}{c} \rtw \in \Dw: \\ \DEwH(\rtw)=\delta\end{array}} \Delta Q(\rtw) \nonumber \\
  d_+^{(Q)}(\delta,\wb,\hat{H}) &=&   \max_{\scriptsize\begin{array}{c} \rtw \in \Dw: \\ \DEwH(\rtw)=\delta\end{array}} \Delta Q(\rtw)\,.
\end{eqnarray}
Hence, it will be the challenge in the following to execute by analytical means the respective constrained minimization and maximization in \eqref{dbounds} for all three relevant physical quantities. Moreover, it will become clear below what `$\delta$ sufficiently small' will mean. Finally, the crucial question will be how the bounds $d_\pm(\delta,\wb,\hat{H})$ depend on $\delta$. To anticipate our results, we are actually going to find that the bounds for all three relevant quantities depend \emph{linearly} on the ensemble error $\Delta \Ew = \delta$. This in turn will confirm the predictive power of the GOK variational principle.

To execute analytically the optimizations in \eqref{dbounds} we first need to understand how various relevant error quantities depend on $\rtw$ and $\hat{U}$ in \eqref{rhoU}, respectively. By using $\rtw \equiv \rtw(\hat{U})$, $\mathbf{X}\equiv (|U_{kl}|^2)_{k,l=0}^{D-1}$, $U_{kl} \equiv \langle \Psi_k|\hat{U}|\Psi_l\rangle = \langle \Psi_k|\tilde{\Psi}_l\rangle$, $\wtb\equiv \mathbf{X}\wb$ and $\Ebt \equiv \mathbf{X}^{\mathrm{T}}\Eb$, rather straightforward calculations yield (one just needs to evaluate various traces in the eigenbasis of either  $\rtw$ or $\rw$ and $\hat{H}$, respectively)
\begin{eqnarray}
    \DEwH(\rtw) &\equiv& \mathrm{Tr}[(\tilde{\rho}_{{\wb}}-\rho_{{\wb}})\hat{H}] \nonumber \\
    &=&   \wb\bd{\cdot}(\mathbf{X}^{\mathrm{T}}\!-\openone)\mathbf{E} \nonumber \\
    &=& \wb\bd{\cdot}(\Ebt-\Eb)  \nonumber \\
    &=& (\wtb-\wb)\bd{\cdot} \Eb     \label{eqn:deltaEw}
\end{eqnarray}
\begin{eqnarray}
    \Delta\rho_{{\wb}}(\rtw) &\equiv& \|\rtw-\rw\|_{\mathrm{HS}}^2 \nonumber \\
    &\equiv& \mathrm{Tr}[(\rtw-\rw)^2] \nonumber \\
    &=& 2{\wb}\bd{\cdot}(\wb-\mathbf{X}\wb) \nonumber  \\
    &\equiv& 2{\wb}\bd{\cdot}({\wb}-\tilde{{\wb}}),\label{eqn:deltarho}
\end{eqnarray}
\begin{eqnarray}
    \Delta\Psi_k(\rtw)  &\equiv& \frac{1}{2}\big\|\ket{\tilde{\Psi}_k}\!\bra{\tilde{\Psi}_k}-\ket{\Psi_k}\!\bra{\Psi_k}\big\|_{\mathrm{HS}}^2\nonumber \\
    &=& 1-|\langle\Psi_k|\tilde{\Psi}_k\rangle|^2 \nonumber \\
    &=& 1-X_{kk}, \label{eqn:deltaPsik}
\end{eqnarray}
\begin{eqnarray}
    \Delta E_k(\rtw) &\equiv& \mathrm{Tr}[\left(\ket{\tilde{\Psi}_k}\!\bra{\tilde{\Psi}_k}-\ket{\Psi_k}\!\bra{\Psi_k}\right) \hat{H}] \nonumber \\
     &=& (\mathbf{X}^{\mathrm{T}}\Eb)_k-E_k \nonumber \\
     &\equiv& \tilde{E}_k-E_k, \label{eqn:deltaEk}
\end{eqnarray}

The far-reaching observation here is that all relevant quantities \eqref{eqn:deltaEw}-\eqref{eqn:deltaEk} can be expressed as \emph{linear} functions of either the matrix $\mathbf{X}\equiv (|U_{kl}|^2)_{k,l=0}^{D-1}$, or the even simpler vectors  $\wtb=\mathbf{X}{\wb}$ and $\Ebt=\mathbf{X}^{\rm T}\mathbf{E}$, respectively. Accordingly, in order to exploit this simplifying structure we need to understand better the sets of such matrices $\mathbf{X}$ and such vectors $\wtb$ and $\Ebt$. This is accomplished in the following section.

\subsection{Birkhoff polytopes and permutohedra} \label{sec:poly}


In the spirit of the last comment in the previous section, the problem \eqref{dbounds} of ensemble minimization and maximization can be conveniently cast in the language of Birkhoff polytopes (and their projections, permutohedra). Birkhoff polytopes are central mathematical objects that have been widely employed in various fields of physics (see ~\cite{bengtsson2005birkhoff} and references therein). The rich mathematical structures of these polytopes will turn out to be instrumental for the derivation of our main results and the comprehensive solution of the scientific problem (see p.~4).

The so-called Birkhoff polytope $\mathcal{B}_D$ of order $D$ is a convex polytope consisting of $D$-by-$D$ doubly stochastic matrices~\cite{marshall1979inequalities}.
A matrix $\mathbf{X}$ is called doubly stochastic (or bistochastic) if its columns and rows consist of real, non-negative entries which for each row and column sum to 1,
\begin{equation}
    \sum_{l=0}^{D-1}X_{lk} = \sum_{l=0}^{D-1}X_{kl} = 1, \quad \forall k=0,1,\ldots, D-1.
\end{equation}
Additionally, if there exists a unitary matrix $\mathbf{U}$ such that $X_{kl} = |U_{kl}|^2$, then $\mathbf{X}$ is called unistochastic.
We denote the set of unistochastic matrices as $\mathcal{U}_D$ which according to Eqs.~\eqref{eqn:deltaEw}-\eqref{eqn:deltaEk} will play a key role in the solution of the scientific problem. For $D\geq 3$, it turns out that the set $\mathcal{U}_D$ is no longer convex~\cite{marshall1979inequalities,bengtsson2005birkhoff,Zyczk09}. In particular, there exist doubly stochastic matrices that are not unistochastic, and therefore $\mathcal{U}_D$ is then a strict subset of $\mathcal{B}_D$.

In order to develop a better intuition for the Birkhoff polytope, we recall in particular the Birkhoff-von Neumann theorem~\cite{birkhoff1946three,von1953certain}: The Birkhoff polytope $\mathcal{B}_D$ is a convex polytope in $\RR^{D\times D}$ of dimension $(D-1)^2$ whose extremal points are the permutation matrices~\cite{marshall1979inequalities,marcus1992survey}. The latter are precisely those matrices which contain in each row and column one `1' and $D-1$ `0'-entries, and together they form a discrete set with $D!$ elements denoted by $\mathcal{S}_D$. Accordingly, by definition, $\mathcal{B}_D$ follows as the convex hull $\mathrm{Conv}(\mathcal{S}_D)$. Moreover, two extremal points (vertices) of $\mathcal{B}_D$ are connected by an edge if and only if they differ by a single cycle~\cite{brualdi1977convex}. Since all permutation matrices are unistochastic, the convex hull $\mathrm{Conv}(\mathcal{U}_D)$ is precisely $\mathcal{B}_D$. The relations among $\mathcal{B}_D$, $\mathcal{U}_D$ and $\mathcal{S}_D$ are depicted in the left panel of Figure \ref{fig:birkhoff}.
They can be summarized as
\begin{equation}
    \mathcal{S}_D \subseteq \mathcal{U}_D \subseteq \mathcal{B}_D, \label{eqn:birkhoff1}
\end{equation}
and
\begin{equation}
    \mathcal{B}_D = \mathrm{Conv}(\mathcal{U}_D) = \mathrm{Conv}(\mathcal{S}_D). \label{eqn:birkhoff2}
\end{equation}

Since we need to minimize and maximize the linear functions \eqref{eqn:deltarho}-\eqref{eqn:deltaEk} over a compact convex set according to Eq.~\eqref{dbounds}, we discuss in the following some basics of linear optimization. We focus on minimization, yet analogous conclusions can be drawn also for maximization. First, we recall the general fact that real-valued \emph{linear} functions $L(\mathbf{X})$  on compact sets attain their minima at extremal points. In particular, in our context this implies
\begin{equation}\label{eqn:unconstr_eq}
    \begin{split}
        \min_{\mathbf{X}\in \mathcal{U}_D} L(\mathbf{X}) = \min_{\mathbf{X}\in \mathcal{B}_D} L(\mathbf{X}) = \min_{\mathbf{P}\in\mathcal{S}_D} L(\mathbf{P}).
    \end{split}
\end{equation}
This becomes obvious in the geometric picture as it is illustrated in the left panel of Figure \ref{fig:birkhoff}: to minimize $L(\Xb)$ one needs to shift the blue hyperplane defined by $L(\Xb)=\mathrm{const}$ in the direction of its normal vector until the boundary of the underlying set is reached. If we restrict, however, the minimization of $L(\Xb)$ in \eqref{eqn:unconstr_eq} by one (or several) linear constraint $\mathcal{A}(\mathbf{X})=c \in \RR$ (in our case $\DEwH(X)=\delta$), in general only the following inequality holds
\begin{equation} \label{eqn:constr_ineq}
    \begin{split}
        \min_{\mathbf{X}\in \mathcal{U}_D, \, \mathcal{A}(\mathbf{X})=c} L(\mathbf{X}) \geq \min_{\mathbf{X}\in \mathcal{B}_D, \, \mathcal{A}(\mathbf{X})=c} L(\mathbf{X}).
    \end{split}
\end{equation}
This crucial distinction between unconstrained and constrained minimization is illustrated in Figure \ref{fig:birkhoff}.
We can see that in the latter case (right panel), the minimizers of $L(\Xb)$ over $\mathcal{B}_D$ and $\mathcal{U}_D$ do not necessarily coincide anymore.

\begin{figure}[t]
    \centering
    \includegraphics[scale=0.19]{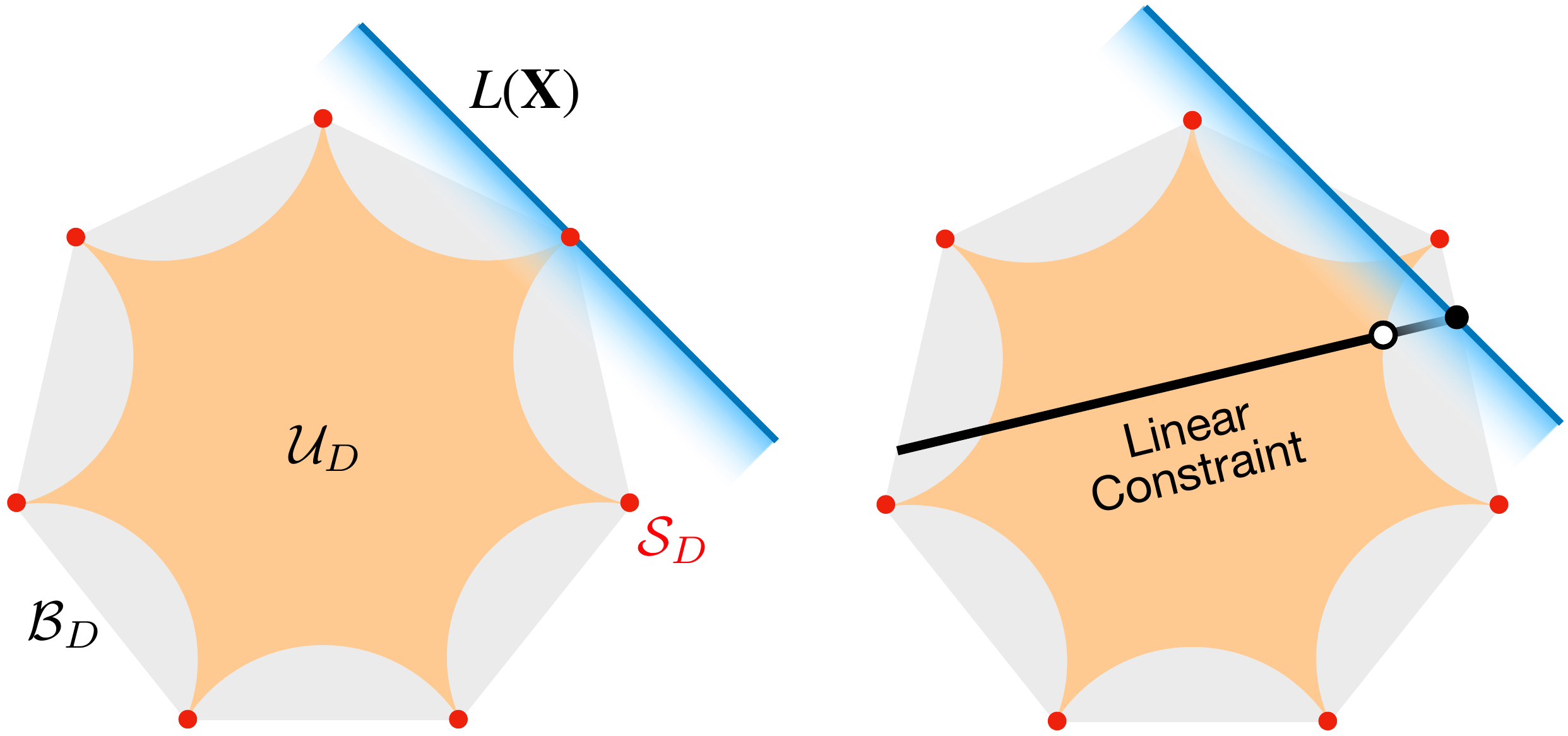}
    \caption{Schematic illustration of unconstrained (left) versus constrained minimization (right) of a linear function $L(\mathbf{X})$ (blue line) over the `grey' Birkhoff polytope $\mathcal{B}_D$. The `orange' subset $\mathcal{U}_D$ of unistochastic matrices and the discrete `red' set $\mathcal{S}_D$ of extremal elements of $\mathcal{B}_D$ (permutation matrices) are also presented.}
    \label{fig:birkhoff}
\end{figure}

In the particular case where $L(\mathbf{X})\equiv l(\mathbf{X}{\wb})$ with some potential constraint $\mathcal{A}(\mathbf{X})\equiv a(\mathbf{X}{\wb})=c$ for some vector $\wb\in \mathbb{R}^D$, the minimization of $L(\Xb)$ over $\mathcal{U}_D$ and $\mathcal{B}_D$, respectively, can be simplified to a minimization of $l(\wtb)$ over the simpler permutohedron $P({\wb})$ of ${\wb}$\footnote{When $\wb$ contains repeated elements, as in some of our cases where some entries of $\wb$ can be zero, $P(\wb)$ is not in a strict sense a permutohedron. When $\wb$ contains repeated elements, $P(\wb)$ no longer has the correct number of vertices and edges.}, given by~\cite{postnikov2009permutohedra}
\begin{equation}\label{perm}
    P({\wb}) \equiv \{\wtb\in\mathbb{R}^D\,|\,\exists \mathbf{X} \in \mathcal{B}_D:\,\wtb=\mathbf{X}{\wb}\}.
\end{equation}
The set $P({\wb})$ inherits the convex properties of $\mathcal{B}_D$, but its dimensionality is much lower.
We present in Figure \ref{fig:b3vsp3} a comparison between the Birkhoff polytope $B_3$,
and the permutohedron of a vector ${\wb}\in \RR^3$. Since the extremal points of $B_3$ are permutation matrices, the extremal elements of $P({\wb})$ follow as the permutations of the vector ${\wb}$. Since all permutations of the entries of a vector in three dimension are cyclic,
all vertices of $B_3$ are connected by an edge. However, in general, two extremal points (vertices) of $P({\wb})$ are connected if and only if they can be transformed into one another by an adjacent transposition $S_{i,i+1}$~\cite{postnikov2009permutohedra}. Consequently, all edges of $P({\wb})$ consist of vectors of the form $\tilde{{\wb}} = \mathbf{X}'{\wb}$, where $\mathbf{X}'$ are unistochastic, and the following convex relaxation is exact
\begin{equation}
    \begin{split}
        \min_{\mathbf{X}\in \mathcal{U}_D,\,\mathcal{A}(\mathbf{X})=c}L(\mathbf{X}) &= \min_{\mathbf{X}\in \mathcal{B}_D,\, \mathcal{A}(\mathbf{X})=c} L(\mathbf{X})
        \\
        &= \min_{\tilde{{\wb}}\in P({\wb}),\,a(\tilde{{\wb}})=c} l(\tilde{{\wb}}).
    \end{split}
\end{equation}
\begin{figure}[t]
    \centering
    \includegraphics[scale=0.145]{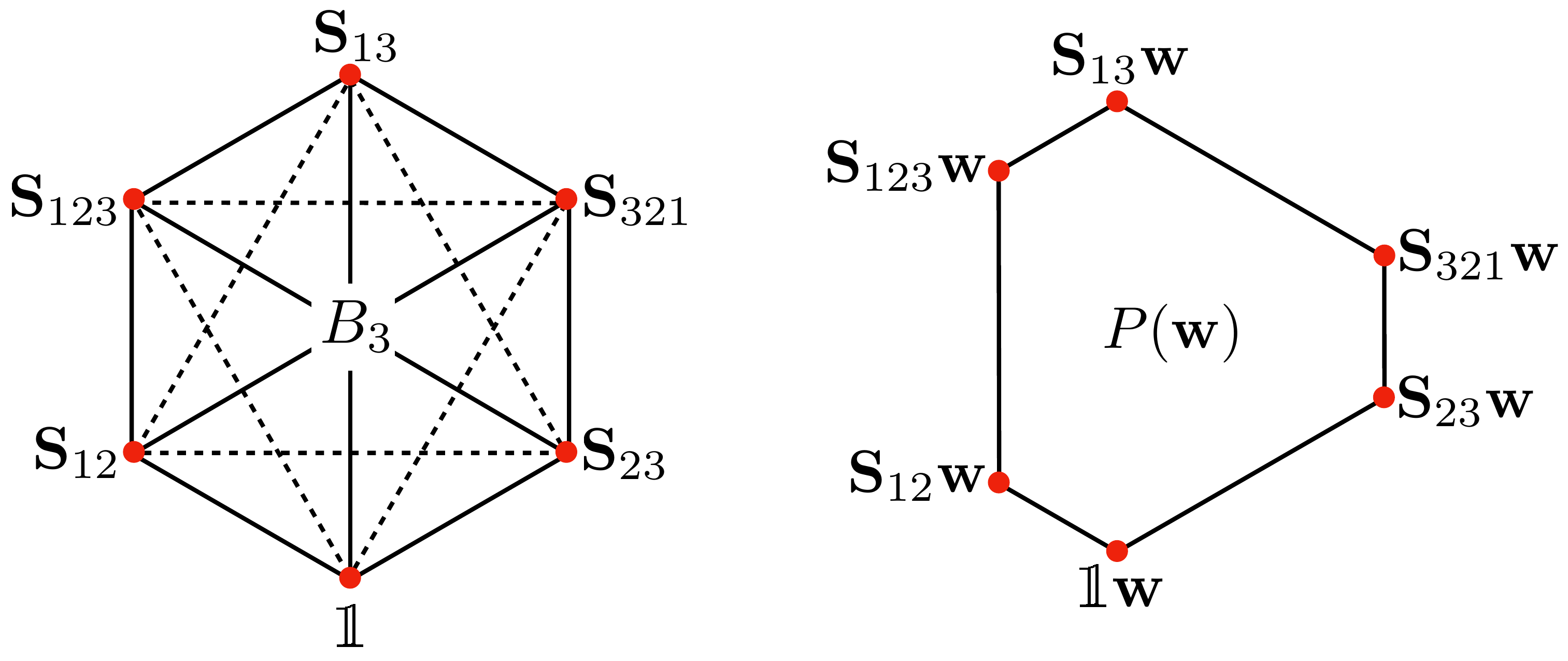}
    \caption{Schematic illustration of the four-dimensional Birkhoff polytope $B_3$ (left) and the two-dimensional permutohedron $P({\wb})$ of a generic vector ${\wb}\in\mathbb{R}^3$ (right). Vertices are presented as red dots and cyclic notation for the six permutations is used. Edges are depicted by solid or dashed lines. On the left panel, solid and dashed lines are used to indicate different lengths of edges. }
    \label{fig:b3vsp3}
\end{figure}

We conclude this section by providing a proof of the GOK variational principle, Theorem \ref{thrm:gok}. Its simplicity, also compared to the original proof in Ref.~~\cite{GOK1988}, already demonstrates how effective the tools from convex analysis are that we have introduced in this section.
\begin{proof}
    We know that $\mathrm{Tr}[\tilde{\rho}_{{\wb}}\hat{H}] = {\wb}^{\mathrm{T}} \mathbf{X}^{\mathrm{T}} \mathbf{E}$ is linear in $\mathbf{X}$, where $\Xb \equiv (|\langle\Psi_k|\tilde{\Psi}_l\rangle|^2)_{k,l=0}^{D-1}\in \mathcal{U}_D$.
   Therefore, according to Eq.~\eqref{eqn:unconstr_eq} and since $\mathcal{U}_D=\mathcal{U}_D^{\mathrm{T}}$ is invariant under transposition, we can relax the problem of minimizing $\mathrm{Tr}[\tilde\rho_{{\wb}}\hat{H}]$ over the set $\mathcal{U}_D$ of unistochastic matrices $\mathbf{X}$ to the set $\mathcal{S}_D$ of permutation matrices $\mathbf{P}$. This yields
    \begin{equation}
        \begin{split}
            \min_{\rtw} \mathrm{Tr}[\tilde\rho_{{\wb}}\hat{H}] &= \min_{\mathbf{X}\in \mathcal{U}_D} {\wb}^{\mathrm{T}} \mathbf{X}^{\mathrm{T}} \mathbf{E}
            \\
            &=\min_{\mathbf{P} \in \mathcal{S}_D} {\wb}^{\mathrm{T}} \mathbf{P} \mathbf{E}
            \\
            &= {\wb}^\downarrow\cdot \mathbf{E}^\uparrow.
        \end{split}
    \end{equation}
In the last step we applied the rearrangement inequality~\cite{hardy1952inequalities}.
\end{proof}

\section{Error Bounds on quantum states and eigenenergies} \label{sec:bounds}

In this section, we present our main results, namely according to the `Scientific Problem' on p.~4 the lower and upper bounds \eqref{thmineq} on the errors of the ensemble state, and the individual eigenstates and eigenenergies as function of the ensemble error \eqref{rhoerror}.

First, we notice that the condition of a fixed error $\DEwH= \delta > 0$ in the ensemble energy defines a linear constraint, which is geometrically
represented as an hyperplane $\mathbb{A}(\delta)$ in $\mathcal{B}_D$, $P({\wb})$, and $P(\mathbf{E})$.
Consequently, there is a straightforward strategy for determining the lower and upper bounds of the \emph{linear} functions \eqref{eqn:deltarho}-\eqref{eqn:deltaEk} according to \eqref{dbounds}: Following Sec.~\ref{sec:poly}, one just needs to evaluate these functions at the vertices of the convex polytope obtained by intersecting $\mathbb{A}(\delta)$ with the corresponding original polytope $\mathcal{B}_D$, $P({\wb})$, and $P(\mathbf{E})$, respectively. The lowest/highest value that is found is nothing else than the sought-after bound $d_{\pm}^{(Q)}(\delta,\wb,\hat{H})$ for the corresponding quantity $Q$.
\begin{figure}[t]
    \centering
    \includegraphics[scale=0.21]{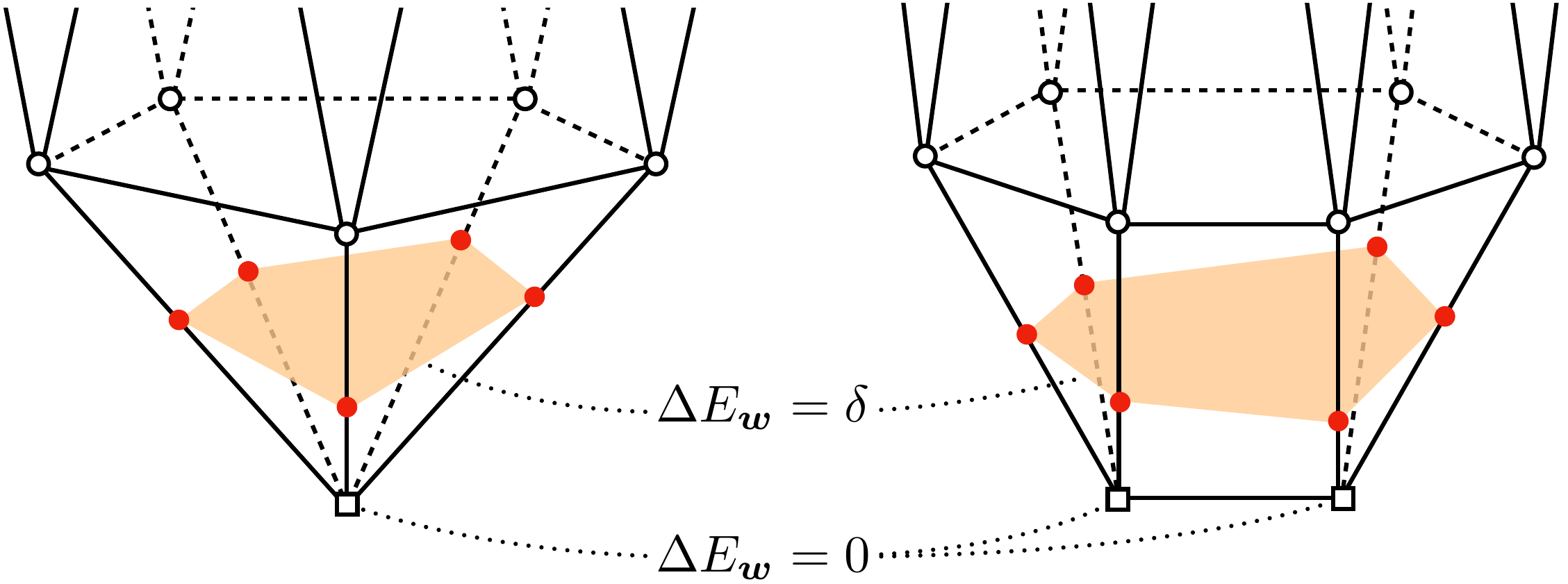}
    \caption{Linear constrained optimization on a polytope. Reference vertices ($\Delta E_{\wb}=0$) are represented as squares. The convex constraint domain ($\Delta{E}_{\wb}=\delta$) is shown as the orange area, whose vertices are red circles.
    On the left panel there is only one reference vertex. On the right panel we show the case when there are multiple reference vertices.}
    \label{fig:constraint_opt}
\end{figure}

In Figure \ref{fig:constraint_opt} we illustrate our strategy for determining the bounds \eqref{thmineq} according to \eqref{dbounds}.
The hyperplane $\mathbb{A}(\delta)$ defined by $\DEwH = \delta > 0$ intersects with the underlying polytope at the orange region.
Since according to the GOK variational principle the ensemble error $\DEwH$ is always non-negative we can distinguish the vertices of
of $\mathcal{B}_D$ and $P(\wb)$ and $P(\Eb)$, respectively, into two categories: The so-called reference vertices characterized by a vanishing error, $\DEwH =0$, and the other ones whose error is finite, $\DEwH >0$. Since polytopes have finitely many vertices, we can always find a $\delta$ sufficiently small such that only reference vertices obey $\DEwH  < \delta$.
When all entries in ${\wb}$ are distinct, there will always be just one reference vertex: In the Birkhoff polytope $\mathcal{B}_D$ it is the  identity matrix and in the permutohedra $P(\wb)$ and $P(\Eb)$ the untransformed vectors ${\wb}$ and $\mathbf{E}$, respectively.
When only $K < D$ entries of ${\wb}$ are positive and distinct (and the rest are zero), there will be several reference vertices.
To understand this, we point out that permuting the energy levels above the $K$-th eigenstate does not change the ensemble energy.

By anticipating some insights from the following derivations, the criterion of $\delta$ sufficiently small can be made precise. For this, and for later purposes we introduce two auxiliary functions:
\begin{eqnarray} \label{eqn:gG}
    g_{\wb,\Eb}&=&\min_{\substack{k<D-1\\w_k>w_{k+1}}}(w^\downarrow_k\!-\!w^\downarrow_{k+1})(E^\uparrow_{k+1}\!-\!E^\uparrow_k) \nonumber
    \\
    G_{\wb,\Eb}&=&(w^\downarrow_0-w^\downarrow_{D-1})(E^\uparrow_{D-1}-E^\uparrow_0).
\end{eqnarray}
The quantities $g_{\wb,\mathbf{E}}$ and $G_{\wb,\mathbf{E}}$ are precisely the minimal and maximal errors in the ensemble energy $E_{\wb}$ when two eigenstates in the exact ensemble $\rw$ are swapped.
Moreover, our assumption of $\delta$ sufficiently small translates into the condition $\delta \leq g_{{\wb},\mathbf{E}}$.
In that practically relevant regime, the derived bounds \eqref{dbounds} will simplify to $d_\pm^{(Q)}(\delta,\wb,\hat{H}) = d_\pm^{(Q)}(\wb,\hat{H}) \,\delta$, as a consequence of the linearity of the constraint function \eqref{eqn:deltaEw} and target functions \eqref{eqn:deltarho}-\eqref{eqn:deltaEk}.

The procedure for determining $d_\pm^{(Q)}(\wb,\hat{H})$ consists of three steps: (i) Identification of all edges of the underlying polytope $\mathcal{B}_D$, $P(\wb)$ or $P(\Eb)$ that intersects with $\mathbb{A}(\delta)$, (ii) computation of the vertices of the intersection (red vertices in Figure \ref{fig:constraint_opt}), and (iii) evaluation of the target function on all vertices of the intersection polytope and identification of its minimal and maximal value. In the following sections, we will determine $d_\pm^{(Q)}(\wb,\hat{H})$ for two important classes of $\wb$: (a) $\wb$ with $w_0 \!>\! w_1 \!>\! \cdots \!>\! w_{D-1}\geq 0$ which can be employed in practice only for smaller systems in order to compute the \emph{entire} energy spectrum, and for larger systems (b) $\wb$ with $w_0 \!>\! w_1 \!>\! \cdots \!>\! w_{K} \!=\! w_{K+1} \!=\! \cdots \!=\! w_{D-1} \!=\! 0$ for targeting the lowest $K$ eigenstates.

\subsection{Targeting the ensemble state} \label{sec:rho_accu}
Acting in accordance with the `Scientific Problem' on p.~4, the first quantity whose error range we quantify as function of the ensemble error $\DEwH$ is the ensemble state $\rw$. This leads to the GOK-analogue of the well-known estimate \eqref{eqn:RR_state} for the Rayleigh-Ritz principle:
\begin{thrm} \label{res:ens_state_err}
  Consider the `Scientific Problem' for a setting $(\mathcal{H},\hat{H},\wb)$ and recall \eqref{eqn:deltaEw}, \eqref{eqn:deltarho}. The errors $\deltarhow(\rtw)$ of the ensemble state and $\DEwH(\rtw)$ of the ensemble energy are universally related for all $\wb$-ensembles $\rtw$ according to
  \begin{equation}
        a_-(\wb,\Eb) \Delta{E}_{{\wb}} \leq \Delta\rho_{{\wb}} \leq  a_+(\wb,\Eb) \Delta{E}_{{\wb}},
    \end{equation}
    where for $w_0\!>\!w_1\!>\!\cdots\!>\!w_{D-1}$
    \begin{eqnarray}
        a_-(\wb,\Eb) &\equiv& 2\min_{k<D-1} \frac{w_k\!-\!w_{k+1}}{E_{k+1}\!-\!E_{k}} \nonumber \\
        a_+(\wb,\Eb) &\equiv& 2\max_{k<D-1} \frac{w_k\!-\!w_{k+1}}{E_{k+1}\!-\!E_{k}},
    \end{eqnarray}
    and for $w_0\!>\!w_1\!>\!\cdots\!>\!w_{K}\!=\!\cdots\!=\!w_{D-1}=0$ ($K \!<\! D\!-\!1$)
       \begin{equation}
       \begin{split}
        a_-(\wb,\Eb) \!&\equiv\! 2\min\Big\{\!\min_{k<K-1} \frac{w_k\!-\!w_{k+1}}{E_{k+1}\!-\!E_{k}}, 
        \\
        &\quad\quad\quad\quad\quad\frac{w_{K-1}}{E_{D-1}\!-\!E_{K-1}} \!\Big\}\nonumber
        \\
        a_+(\wb,\Eb) \!&\equiv\! 2\max_{k<K} \frac{w_k\!-\!w_{k+1}}{E_{k+1}\!-\!E_{k}}.
        \end{split}
    \end{equation}
\end{thrm}

\textit{Sketch of proof.} As explained in Sec.~\ref{sec:poly}, the specific linear form of \eqref{eqn:deltaEw} and \eqref{eqn:deltarho} allows us to transform the underlying optimization problem \eqref{dbounds} into one on the permutohedron $P({\wb})$.
In the regime where $\Delta{E}_{\wb} = \delta < g_{\wb,\mathbf{E}}$, there is only one reference vertex, namely ${\wb}$.
If ${\wb}$ is strictly decreasing ($K=D$), then we know from Sec.~\ref{sec:poly} that the only neighboring vertices are $\mathbf{S}_{i,i+1}{\wb}$, and the connecting edges are convex combinations between ${\wb}$ and $\mathbf{S}_{i,i+1}{\wb}$.
This is why the lower and upper bounds involve only differences between neighboring entries in ${\wb}$ and $\mathbf{E}$.
From here one can easily work out the intersection between $\Delta{E}_{\wb}=\delta$ and $P({\wb})$.
If ${\wb}$ instead fulfills the second requirement ($K<D-1$), then neighboring vertices to $\wb$ include not only $\mathbf{S}_{i,i+1}{\wb}$, but also $\mathbf{S}_{K\!-\!1,k+1}{\wb}$ where $k \geq K$. The detailed proof of Theorem \ref{res:ens_state_err} is reserved to the Appendix.

\begin{figure}[t]
    \centering
    \includegraphics[scale=0.3]{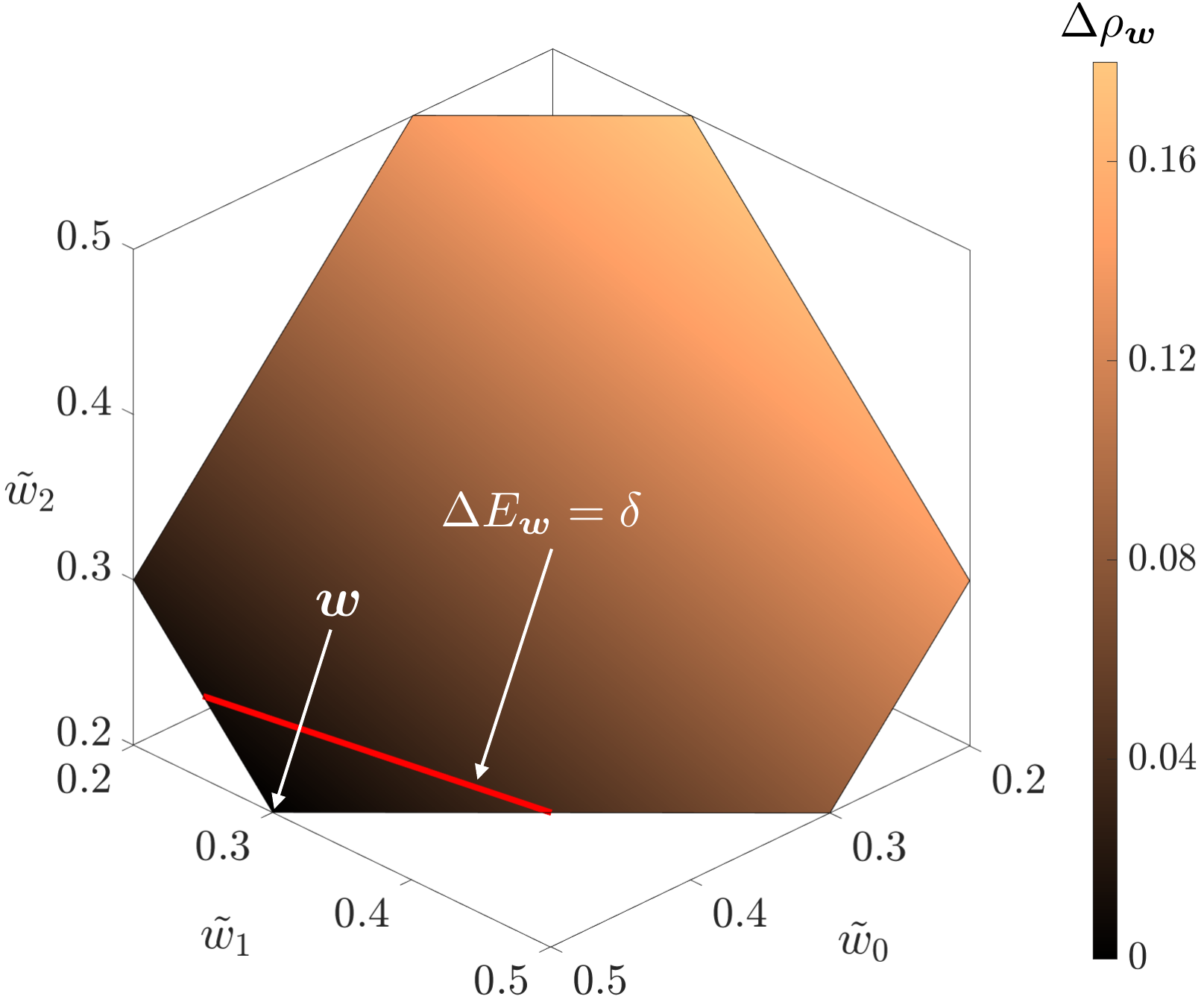}
    \caption{Permutohedron $P({\wb})$ of ${\wb}=(0.5,0.3,0.2)$. The face color encodes the value of $\Delta\rho_{\wb}(\tilde{{\wb}})$  according to \eqref{eqn:deltarho}.
    The red line is the intersection between $P({\wb})$ and the hyperplane defined by $\Delta{E}_{\wb}(\tilde{{\wb}})=\delta = 0.1$ where $\mathbf{E}=(-1,0,2)$.}
    \label{fig:permutahedron}
\end{figure}

We illustrate in Figure \ref{fig:permutahedron} the proof of Theorem \ref{res:ens_state_err} with an example Hamiltonian in three dimensions.
For the concrete plot, we chose the energy spectrum $\mathbf{E}=(-1,0,2)$ and the ensemble weights ${\wb} = (0.5,0.3,0.2)$.
The hexagonal region in Figure \ref{fig:permutahedron} is the permutohedron $P({\wb})$ with elements $\tilde{{\wb}}$.
The red solid segment represents the intersection between $P({\wb})$ and the constraint $\Delta{E}_{{\wb}} = \delta = 0.1$.
The error of the ensemble state $\Delta{\rho}_{\wb}$, which is encoded in color, has a global minimum at the reference vertex ${\wb}$ (bottom left corner).
The minimum and maximum of $\Delta\rho_{\wb}$ within the red segment are obtained at the two end points, and the inequalities provided by Theorem \ref{res:ens_state_err} are explicitly saturated.

Knowing the accuracy of the ensemble state is of central importance, as it allows us to make reliable prediction on the accuracy of the expectation value of any observable $\hat{A}$ with respect to the optimized state.
\begin{coro}
    The error $\Delta A_{{\wb}} \equiv \mathrm{Tr}[(\tilde{\rho}_{{\wb}}-\rho_{{\wb}})\hat{A}]$  of the expectation of an observable $\hat{A}$ satisfies for all $\wb$-ensemble states $\rtw$
    \begin{equation}\label{obsAbound}
        |\Delta A_{{\wb}}| \leq \|\hat{A}\|_{\mathrm{HS}} \sqrt{a_+(\wb,\Eb) \Delta{E}_{{\wb}} }.
    \end{equation}
\end{coro}
\begin{proof}
    Since $\langle \mathbf{X}, \mathbf{Y}\rangle_\mathrm{HS}\equiv \mathrm{Tr}[\mathbf{X}^\dagger\mathbf{Y}]$ defines an inner product, the Cauchy-Schwarz inequality applies and yields $|\mathrm{Tr}[\mathbf{X}^\dagger\mathbf{Y}]|\leq \|\mathbf{X}\|_\mathrm{HS}\|\mathbf{Y}\|_{\mathrm{HS}}$, where $\|\mathbf{X}\|_\mathrm{HS} \equiv \sqrt{\langle \mathbf{X}, \mathbf{X}\rangle_\mathrm{HS}}$.
     With $\mathbf{X}=\hat{A}$ and $\mathbf{Y} = \tilde{\rho}_{\wb}-\rho_{\wb}$, \eqref{obsAbound} then follows directly from Theorem \ref{res:ens_state_err}.
\end{proof}

\subsection{Targeting the eigenstates} \label{sec:Psik_accu}
Next, we quantify the range of the errors of the individual eigenstates.
We remark that for $k\geq K$ the error $\Delta\Psi_k$ cannot be estimated by $\Delta{E}_{\wb}$ anymore since $\ket{\tilde{\Psi}_k}$ does not enter the ensemble $\rtw$ as a consequence of $w_{k\geq K}=0$.
For example, $\Delta\Psi_K$ can take any value between 0 and 1 even if $\Delta{E}_{\wb}=0$, as the energy of the $K$-th eigenstate does not contribute to the ensemble energy.
Conversely, in order to extract accurate information about the target excited states,
the corresponding weights must be non-zero and distinct from others.

\begin{thrm}\label{res:psik_error}
  Consider the `Scientific Problem' for a setting $(\mathcal{H},\hat{H},\wb)$ and recall \eqref{eqn:deltaEw}, \eqref{eqn:deltaPsik}. For any $k$ with non-degenerate $w_{k}$, the error $\deltaPsik(\rtw)$ of the $k$-th eigenstate and $\DEwH(\rtw)$ of the ensemble energy are universally related for all $\wb$-ensembles $\rtw$ according to
  \begin{equation}
        0 \leq \deltaPsik \leq  b_+^{(k)}(\wb,\Eb) \Delta{E}_{{\wb}},
    \end{equation}
    where
    \begin{eqnarray}
        b_+^{(k)}(\wb,\Eb) &\equiv&
        \begin{cases}
            \frac{1}{(w_0-w_1)(E_1-E_0)}, \quad &k = 0
            \\
            \frac{1}{\min\{t_{k-1},t_k\}}\,,\quad &k>0
        \end{cases}
    \end{eqnarray}
     and $t_k \equiv (w_k\!-\!w_{k+1})(E_{k+1}\!-\!E_{k})$.
\end{thrm}

Theorem \ref{res:psik_error} would directly yield a lower and upper bound on the sum $\deltaPsi= \sum_{k=0}^{K-1} \deltaPsik$ of various individual errors $\deltaPsik$, yet they will not be tight. Instead, the optimal bounds on $\deltaPsi$ is given by the following theorem. To formulate it in a compact way and for subsequent considerations we recall the Heaviside step function $\Theta(x)$ defined as
\begin{equation}\label{Heaviside}
    \Theta(x) = \begin{cases}
        1, \quad &x>0
        \\
        1/2,\quad &x=0
        \\
        0, \quad &x < 0
    \end{cases}.
\end{equation}

\begin{thrm}\label{res:eigenstate_error}
Consider the `Scientific Problem' for a setting $(\mathcal{H},\hat{H},\wb)$ and recall \eqref{eqn:deltaEw}, \eqref{eqn:deltaPsik}. Then, the sum 
of the errors of the targeted eigenstates and the error $\DEwH(\rtw)$ of the ensemble energy are universally related for all $\wb$-ensembles $\rtw$ in the following way: When $w_0 \!>\! w_1 \!>\! \cdots \!>\!w_{D-1}$, one has
\begin{equation}
    \begin{split}
    \frac{2\Delta{E}_{\bm{w}}}{G_{\bm{w},\mathbf{E}}}  \leq \sum_{k=0}^{D-1}\Delta\Psi_k \leq \frac{2\Delta{E}_{\bm{w}}}{g_{\bm{w},\mathbf{E}}},
    \end{split}
\end{equation}
where $g_{\bm{w},\mathbf{E}}$ and $G_{\bm{w},\mathbf{E}}$ are defined in \eqref{eqn:gG}.
When $w_0\!>\!w_1\!>\!\cdots\!>\!w_{K} \!=\! \cdots \!=\!  w_{D-1}=0$ for some $K\!<\!D\!-\!1$ one has
\begin{equation}
\begin{split}
    \frac{\Delta{E}_{\bm{w}}}{G_{\bm{w},\mathbf{E}}} \!&\leq\! \sum_{k=0}^{K-1}\!\Delta\!\Psi_k 
    \\
    \!&\leq\!
    \max_{k<K}\! \left\{\!\frac{2\,\Theta(K\!-\!k\!-\!1)}{(w_{k}\!-\!w_{k+1})(E_{k+1}\!-\!E_k)}\!\right\}\Delta{E}_{\wb}.
\end{split}
\end{equation}
\end{thrm}

Since the error of an eigenstate given by \eqref{eqn:deltaPsik} cannot be interpreted as a function of $\Xb \wb$ or $\Xb^{\mathrm{T}}\Eb$, the proofs of Theorems \ref{res:psik_error} and \ref{res:eigenstate_error} will refer to the more involved Birkhoff polytope $\mathcal{B}_D$ rather than the simple permutohedron. Yet, both proofs are conceptually still quite similar to the proof of Theorem \ref{res:ens_state_err}, and can be found in the Appendix.

\subsection{Target Eigenenergies}

Finally, we investigate the bounds for the error $\Delta{E}_k$ of the individual energies.
Unlike other types of errors we have seen before, $\Delta{E}_k$ can be negative, if $k > 0$.
The reason for this also reflects the fact that there is no variational principle for individual excited states.
Nonetheless, according to \eqref{eqn:deltaEk}, $\Delta{E}_k$ is a linear function of $\tilde{\mathbf{E}}$ on the permutohedron $P(\mathbf{E})$.
This means that we can use the same strategies as above to determine its lower and upper bounds.

\begin{thrm} \label{res:ek_error}
    Consider the `Scientific Problem' for a setting $(\mathcal{H},\hat{H},\wb)$ and recall \eqref{eqn:deltaEw}, \eqref{eqn:deltaEk}.
    For any $k$ with non-degenerate $w_{k}$, the error $\deltaEk(\rtw)$ of the $k$-th eigenenergy and $\DEwH(\rtw)$ of the ensemble energy are universally related for all $\wb$-ensembles $\rtw$ according to
  \begin{equation}
        c_-^{(k)}(\wb,\Eb)\Delta{E}_{{\wb}} \leq \deltaEk \leq  c_+^{(k)}(\wb,\Eb) \Delta{E}_{{\wb}},
    \end{equation}
    where
    \begin{eqnarray}
        c_-^{(k)}(\wb,\Eb) &\equiv& \begin{cases}
            0, \quad &k = 0
            \\
            \frac{1}{w_{k}-w_{k-1}}, \quad &k > 0
        \end{cases}
        \\
        c_+^{(k)}(\wb,\Eb) &\equiv& \frac{1}{w_{k}-w_{k+1}}.
    \end{eqnarray}
\end{thrm}

The interpretation of these bounds is also clear. The lower bound, which is negative, is saturated when the $(k\!-\!1)$-th and $k$-th energy levels mix with one another,
and the positive upper bound is realized by mixing the $k$-th and $(k\!+\!1)$-th levels.

The lack of a variational principle for individual excited states results in a potential error cancellation which could affect the accumulated error of all eigenenergies of interest: If one naively sums up all $\Delta{E}_k$, the collective error would be grossly underestimated.
To circumvent such a misleading error cancellation, one should consider instead the absolute sum $\sum_{k=0}^{D-1}|\Delta{E}_k|$ of various individual errors $\Delta{E}_k$.
By directly applying Theorem \ref{res:ek_error} $D$ times, we obtain (while assuming a strictly decreasing ${\wb}$)
\begin{equation}
    0\leq \sum_{k=0}^{D-1}|\Delta{E}_k| \leq \frac{D}{\min_{k<D-1}(w_k\!-\!w_{k+1})} \Delta{E}_{\wb}.
\end{equation}
However, much tighter bounds can be derived without invoking Theorem \ref{res:ek_error}, or the polytope argument.

\begin{thrm} \label{res:en_error}
    Consider the `Scientific Problem' for a setting $(\mathcal{H},\hat{H},\wb)$ and recall \eqref{eqn:deltaEw}, \eqref{eqn:deltaEk}. Then, the sum
    of the errors of the targeted eigenenergies and the error $\DEwH(\rtw)$ of the ensemble energy are universally related for all $\wb$-ensembles $\rtw$ in the following way.
    When $w_0\!>\!w_1\!>\!\cdots\!>\!w_{D-1}$, one has
    \begin{equation}
    \begin{split}
        \frac{2\Delta{E}_{{\wb}}}{w_0\!-\!w_{D-1}}&\leq\sum_{k=0}^{D-1} |\Delta{E}_k| 
        \\
        &\leq \max_{k<D-1}\! \left\{\!\frac{2}{w_{k}\!-\!w_{k+1}}\!\right\}\Delta{E}_{\wb}.
    \end{split}
    \end{equation}
    When $w_{0} \!>\! w_1 \!>\! \cdots \!>\! w_K \!=\! \cdots \!=\! w_{D-1}=0$ for some $K\!<\!D\!-\!1$, one has
    \begin{equation}
    \begin{split}
        \frac{\Delta{E}_{\wb}}{w_0}&\leq \sum_{k=0}^{K-1} |\Delta{E}_k| 
        \\
        &\leq \max_{k<K}\! \left\{\!\frac{2\,\Theta(K\!-\!k\!-\!1)}{w_{k}\!-\!w_{k+1}}\!\right\}\Delta{E}_{\wb}.
    \end{split}
    \end{equation}
\end{thrm}
Theorem \ref{res:en_error} explicitly demonstrated that the absolute sum of the errors of all targeted eigenenergies is bounded from below and above by the error of the ensemble energy. Remarkably, the prefactors of various bounds involve only the weight vector ${\wb}$, and are independent of the spectrum $\Eb$ of the target Hamiltonian. The proof of Theorem \ref{res:en_error} can also be found in the Appendix.

\begin{figure*}[t]
    \centering
    \includegraphics[scale=0.3]{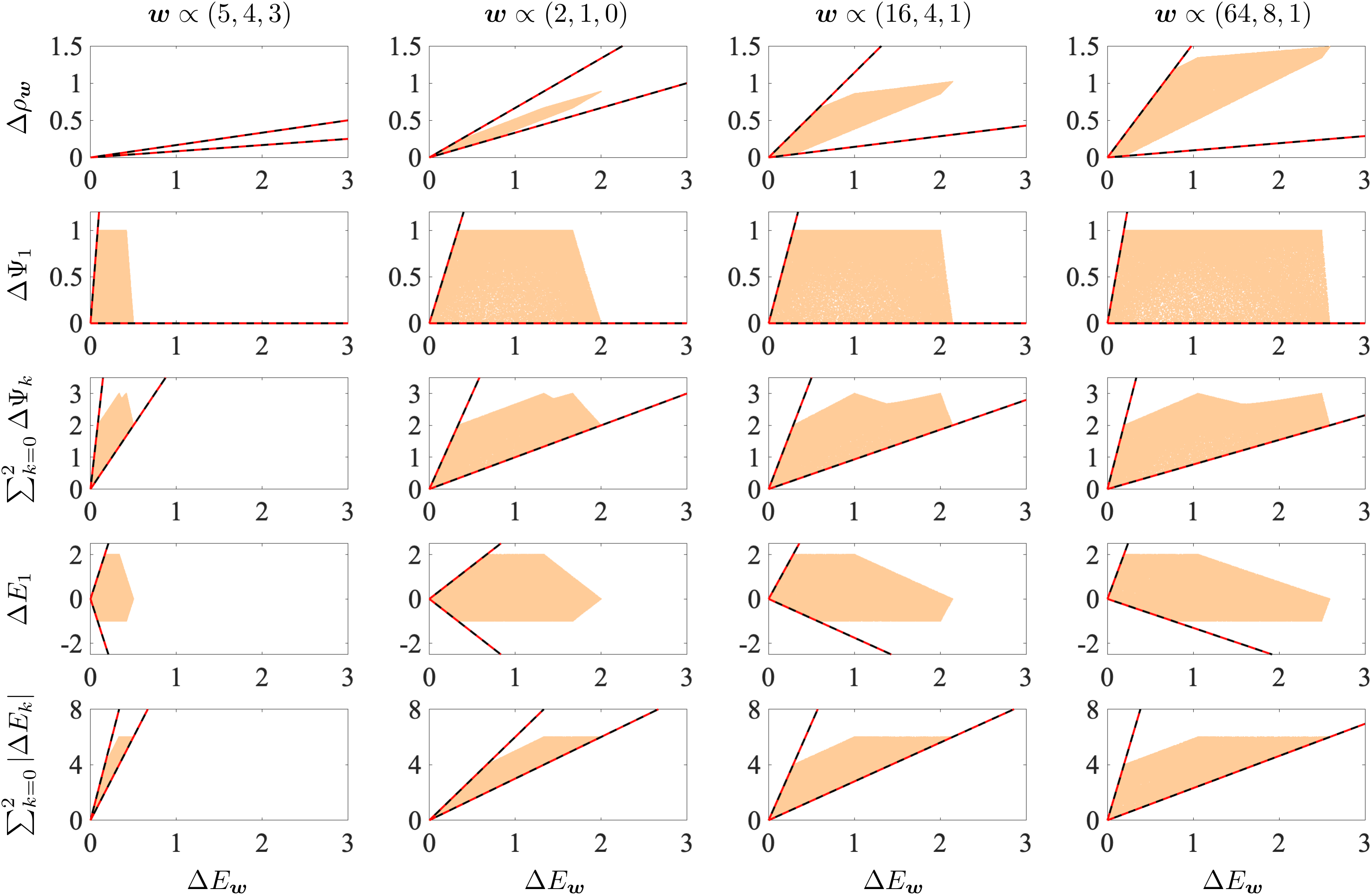}
       \caption{Numerical demonstration of the error bounds predicted by Theorems \ref{res:ens_state_err}~-~\ref{res:en_error} and their tightness. For an energy spectrum $\mathbf{E} = (-1,0,2)$ and four exemplary weight vectors $\wb$, the relations between various errors and the ensemble error $\DEwH$ are presented for $10^5$ randomly sampled $\wb$-ensemble states $\rtw = \hat{U} \rw\hat{U}^\dagger$ (orange points). Theoretical bounds (black solid lines) turn out to coincide with the actual bounds given by the sampled data (red dashed lines).}
    \label{fig:sampling1}
\end{figure*}

\section{Saturation of Inequalities and Optimal choice of ensemble weights} \label{sec:saturation}

In this section we first confirm both analytically and numerically that various error bounds provided by Theorems \ref{res:ens_state_err}~-~\ref{res:en_error} are optimal, in the sense that they can be saturated for any choice of $\wb$, provided  $\DEwH \leq g_{\wb,\Eb}$. Then, based on the explicit form $d_{\pm}\equiv d_{\pm}^{(Q)}(\wb,\Eb)$ of the prefactors of various bounds, we determine how the ensemble weights $\wb$ need to be chosen in order to minimize the error range of the quantum states and energies targeted through the GOK variational principle. In that context, we also explain how these insights can be used in practice to improve numerical approaches in terms of the accuracy of their predictions.

\subsection{Tightness of error bounds}\label{sec:boundstight}
In the following, we demonstrate by analytical means that the error bounds provided by Theorems \ref{res:ens_state_err}~-~\ref{res:en_error} are optimal, i.e., they cannot be replaced by more restrictive bounds. We already know that such tightness is given for the bounds in Theorems \ref{res:ens_state_err}, \ref{res:ek_error}, \ref{res:en_error} as a direct mathematical consequence of their derivation. Quite in contrast, the derivation of the bounds in Theorem \ref{res:psik_error} and \ref{res:eigenstate_error} involved a relaxation of the optimization \eqref{dbounds} from the underlying set $\mathcal{U}_d$ to the larger Birkhoff polytope $\mathcal{B}_d$ according to \eqref{eqn:constr_ineq}.



\begin{figure*}[t]
    \centering
    \includegraphics[scale=0.3]{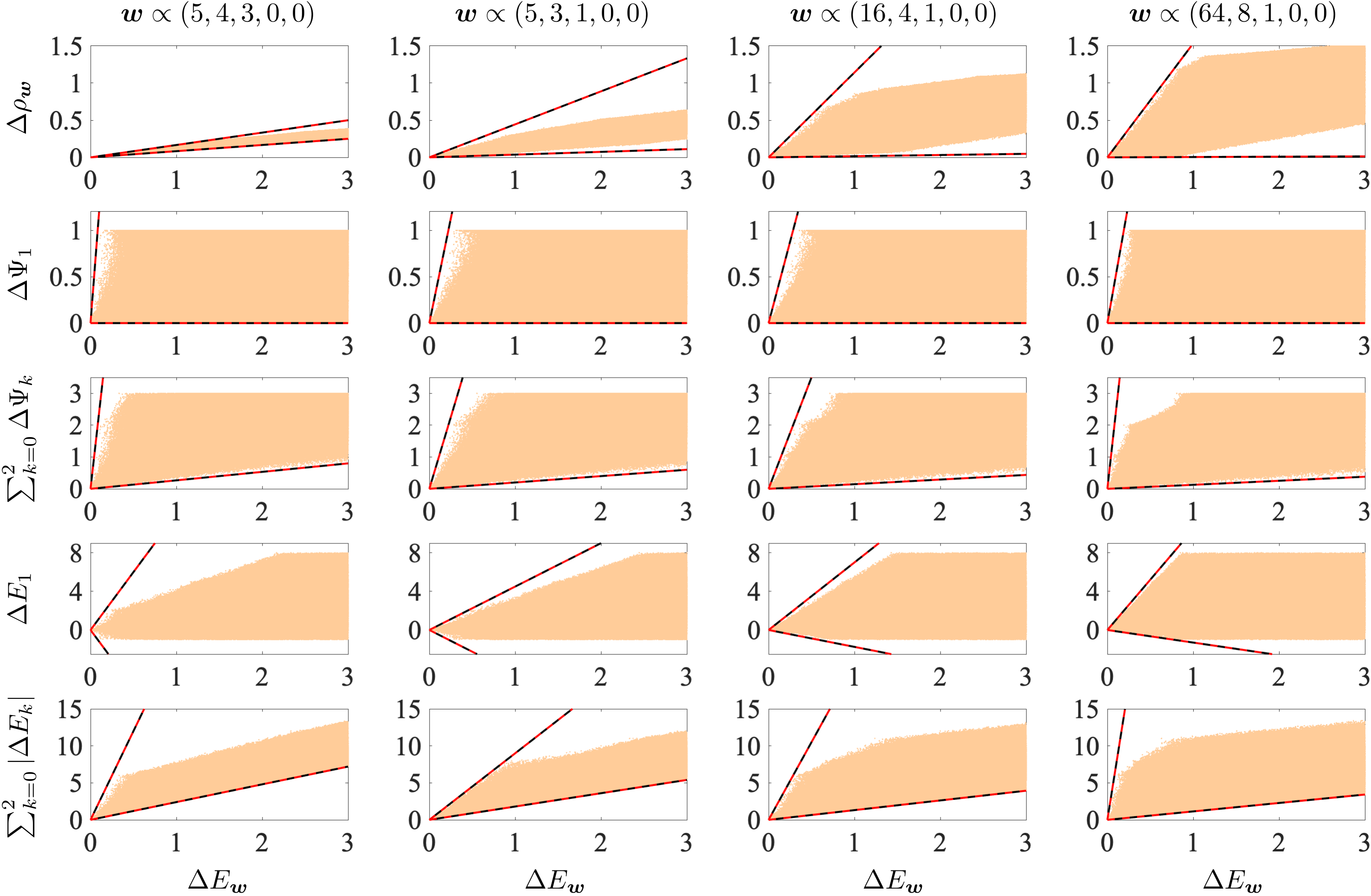}
    \caption{Numerical demonstration of the error bounds predicted by Theorems \ref{res:ens_state_err}~-~\ref{res:en_error} and their tightness. For an energy spectrum $\mathbf{E} = (-1,0,2,5,8)$ and four exemplary weight vectors $\wb$, the relations between various errors and the ensemble error $\DEwH$ are presented for $10^7$ randomly sampled $\wb$-ensemble states $\rtw = \hat{U} \rw\hat{U}^\dagger$ (orange points). Theoretical bounds (black solid lines) turn out to coincide with the actual bounds given by the sampled data (red dashed lines).}
    \label{fig:sampling2}
\end{figure*}

The strategy for confirming tightness is straightforward since we just need to provide for any valid $\wb$ a corresponding $\rtw$ which saturates the corresponding inequality between the errors. In order to illustrate this, we consider as an example the upper bound in Theorem \ref{res:en_error} (despite the fact that its tightness is already known). Let $\wb$ be the underlying non-degenerate weight vector and $k$ the index for which \mbox{$2/(w_l-w_{l+1})$}  is maximized, i.e.,
\begin{equation}
    k = \arg\max_{l<D-1}\left\{\frac{2}{w_l-w_{l+1}}\right\}.
\end{equation}
Then, according to Theorem \ref{res:en_error}, the sum of the absolute errors of various eigenenergies is bounded from above by
\begin{equation}\label{Thm7ineq}
    \sum_{l=0}^{D-1} |\Delta{E}_l| \leq \frac{2}{w_k\!-\!w_{k+1}}\Delta{E}_{\wb}.
\end{equation}
To confirm that this bound can be attained, we construct a trial state $\rtw$ for which this inequality is saturated. Our $\rtw$ is obtained from the exact state $\rw$ by applying the following Jacobi rotation $\mathbf{J}_{k,k+1}(\theta)$
\begin{equation}\label{Jaccobi}
    \begin{split}
        |\Psi_k\rangle &\longmapsto \cos(\theta)|\Psi_k\rangle + \sin(\theta)|\Psi_{k+1}\rangle,
        \\
        |\Psi_{k+1}\rangle &\longmapsto -\sin(\theta)|\Psi_k\rangle + \cos(\theta)|\Psi_{k+1}\rangle,
    \end{split}
\end{equation}
while keeping all other states unchanged. This leads directly to (recall Eq.~\eqref{eqn:deltaEw})
\begin{equation}\label{jacsatur}
        \Delta{E}_{{\wb}} = \sin^2(\theta)(w_k\!-\!w_{k+1})(E_{k+1}\!-\!E_k),
\end{equation}
and
\begin{equation}
        \sum_{l=0}^{D-1}|\Delta{E}_l| = 2\sin^2(\theta)(E_{k+1}\!-\!E_k).
\end{equation}
Comparing these two equations leads to
\begin{equation}
    \sum_{l=0}^{D-1}|\Delta{E}_l| = \frac{2\Delta{E}_{{\wb}}}{w_k-w_{k+1}},
\end{equation}
i.e., the absolute sum of the individual energy errors is proportional to the ensemble error,
exactly by the prefactor stated in Theorem \ref{res:en_error}. Moreover, based on \eqref{jacsatur} and by recalling  \eqref{eqn:gG} we conclude that $\wb$-ensemble states $\rtw$ saturating the bound in Theorem \ref{res:en_error} exist at least as long as $\DEwH \leq (w_k\!-\!w_{k+1})(E_{k+1}\!-\!E_k)$. It is worth noticing that the right-hand side is never smaller than $g_{\wb,\Eb}$ in \eqref{eqn:gG}. Last but not least, it is not difficult to see that such specific `saturating' Jacobi rotations \eqref{Jaccobi} also exist for various other bounds, in particular for those of Theorem \ref{res:psik_error} and \ref{res:eigenstate_error} and also for the second case of weight vectors with $w_0 > \ldots > w_{K-1}>w_K=\ldots w_{D-1}=0$. This eventually confirms that all our error bounds are optimal in the sense that they are mathematically tight, provided the ensemble error $\DEwH$ is small enough, i.e., $\DEwH \leq g_{{\wb},\mathbf{E}}$. The numerical analysis in the subsequent section and the analytical example above reveals that for some error bounds, the range of $\DEwH$ for which the linear bounds are tight extends even beyond $g_{{\wb},\mathbf{E}}$. Yet, from a practical point of view it does not make sense to elaborate more on that point since the relevant regime is the one of reasonably good convergence, i.e., $\DEwH$ sufficiently small.

\subsection{Numerical illustration of error bounds}\label{sec:sampling}

To directly showcase the universality and the tightness of our bounds, we present in Figure \ref{fig:sampling1} and \ref{fig:sampling2}
various errors in the quantum states and eigenenergies,
of randomly sampled eigenbases encoded as unitary matrices.
The unitaries are parameterized as exponentials of antisymmetric matrices,
whose elements are sampled uniformly in the interval $[-\pi,\pi]$. Additionally, all possible permutation matrices are included into the samples.
We test two different scenarios: (a) a Hilbert space of dimension $D=3$ where all three eigenstates are targeted (Figure \ref{fig:sampling1}), and (b) a Hilbert space of dimension $D=5$ where the lowest three eigenstates are targeted (Figure \ref{fig:sampling2}).
For both scenarios, we choose four exemplary weight vectors ${\wb}$: One that has nearly equal entries, the one given by \eqref{eqn:optw1} and \eqref{eqn:optw2} which turns out to minimize the accumulated energy error $\Delta E$, and those characterized by $w_{i}/w_{i+1}=4$, and $w_{i}/w_{i+1}=8$, respectively (excluding zeros).

As the results shown in Figures \ref{fig:sampling1} and \ref{fig:sampling2} clearly reveal, our analytic error bounds are satisfied and can be well saturated for both scenarios.  The latter is evident by the fact that the theoretical bounds (black solid lines) and the statistical bound based on the sampled data (red dashed line) coincide. Yet, we also see that already for the relatively small dimension $D=5$, the sampling of $10^7$ trial states $\rtw =\hat{U} \rw \hat{U}^\dagger$ is not sufficient anymore to fill out the entire area of admissible pairs $(\Delta Q(\rtw), \DEwH(\rtw))$. This is due to the increase of the dimension of the manifold of unitaries as $D^2$, which in turn highlights the value of our analytical analysis in the previous and subsequent section. Moreover, by comparing the numerical results for the four exemplary weight vectors $\wb$, we observe the following trends. First, the range of possible errors $\deltarhow, \deltaPsi, \deltaE$  defined by the respective lower and upper bounds is, the narrower the more equal various weights $w_l$ are. In contrast, for $\deltaPsik, \deltaEk$ the range gets narrower as the weight gaps $|w_k-{w_{k\pm 1}}|$ get bigger. In practical applications, however, the ideal selection criteria for $\wb$ would not refer to the narrowness of the error range but instead to the `flatness' of the upper bound, i.e., the value of the corresponding prefactor. Indeed, the smaller the upper bound is, the smaller must be the error of the quantity of interest and thus the more accurate are the predictions made by the numerical approach. In that regard, for all quantities except the ensemble state, the $\wb$ of the second column, namely an equally spaced $\wb$, seems to be superior. Since the latter observation might be subject to the specific energy spectra chosen here, we will elaborate in the next section more on the crucial point of weight selection. In particular, we will succeed in determining by analytical means the optimal weight vectors for each quantity of interest, and their potential dependence on the underlying energy spectrum $\Eb$.
\begin{table*}
\scriptsize
    \begin{tabular}{|c|c|c|c|c|}
    \hline
    \rule{0pt}{0ex}\rule[-1ex]{0pt}{0pt}
    \textbf{Quantity} & \textbf{Error} & \textbf{Theorem} & \textbf{Lowest Upper Bound}& \textbf{Optimal} $\wb$
    \\
    \hline
    \rule{0pt}{3.5ex}\rule[-5.5ex]{0pt}{0pt}
    $\ket{\Psi_k}$ &$\Delta \Psi_k$ &{\ref{res:psik_error}}
    & $\:k(r_-\!+\!r_+)\!+\!r_+$
    & $\wb \propto (\underbrace{r_{-}\!+\!r_{+},\ldots,r_{-}\!+\!r_{+}}_{k-1},r_{+},\underbrace{0,\ldots,0}_{D-k})$
    \\
    \hline
    \rule{0pt}{4.5ex}\rule[-3ex]{0pt}{0pt}
    $\{\ket{\Psi_k}\}_{k=0}^{D-1}$&$\sum_{k=0}^{D-1}\Delta \Psi_k$ &{\ref{res:eigenstate_error}} &  $\sum_{k=1}^{D-1}2k(E_{k}\!-\!E_{k-1})^{-1}$ & $\:\bm{w}\propto \left(\sum_{l=k+1}^{D-1}(E_{l}\!-\!E_{l-1})^{-1}\right)_{k=0}^{D-1}\:$
    \\
    \hline
    \rule{0pt}{4.5ex}\rule[-3ex]{0pt}{0pt}
    $\{\ket{\Psi_k}\}_{k=0}^{K-1}$&$\sum_{k=0}^{K-1}\Delta \Psi_k$ &{\ref{res:eigenstate_error}} &  $\sum_{k=1}^{K}2k(E_{k}\!-\!E_{k-1})^{-1}\Theta(K\!-\!k)$ & $\bm{w}\propto \left(\sum_{l=k+1}^{D-1}(E_{l}\!-\!E_{l-1})^{-1}\Theta(K\!-\!l)\right)_{k=0}^{D-1}$
    \\
    \hline
    \rule{0pt}{3.5ex}\rule[-5.5ex]{0pt}{0pt}
    $E_k$ & $|\Delta E_k|$ &  \ref{res:ek_error}  & $2k+1$ & $\bm{w}\propto(\underbrace{2,\ldots,2}_{k},1,\underbrace{0,\ldots,0}_{D-k-1})$
    \\
    \hline
    \rule{0pt}{3.5ex}\rule[-2ex]{0pt}{0pt}
    $\{E_k\}_{k=0}^{D-1}$ & $\sum_{k=0}^{D-1}|\Delta E_k|$ & {\ref{res:en_error}}   &  $D(D\!-\!1)$ & $\bm{w}\propto(D\!-\!1,D\!-\!2,\ldots, 1,0)$
    \\
    \hline
    \rule{0pt}{3.5ex}\rule[-2ex]{0pt}{0pt}
    $\{E_k\}_{k=0}^{K-1}$ & $\sum_{k=0}^{K-1}|\Delta E_k|$ & {\ref{res:en_error}}   &  $K^2$ & $\bm{w}\propto(2K\!-\!1,2K\!-\!3,\ldots,1,0,\ldots,0)$
    \\
    \hline
    \end{tabular}
    \caption{The lowest upper bounds and the corresponding optimal choice of $\bm{w}$ (normalized to 1) for various errors of the quantum states and eigenenergies. We defined $r_{\pm} \equiv |E_{k}-E_{k\pm1}|^{-1}$, $\Theta(x)$ denotes the Heaviside step function \eqref{Heaviside}, and $K$ satisfies $K\!<\!D\!-\!1$.}\label{tab:opt_bounds}
\end{table*}

\subsection{Optimal choice of weights}\label{sec:selectionw}
The main results of our work, the explicit expressions of the linear error bounds in Theorems \ref{res:ens_state_err}~-~\ref{res:en_error} allow us to determine in the following the optimal auxiliary weights $w_k$ for practical applications. Since we know from the theoretical and numerical analyses in the previous sections that various error bounds are tight, we just need to minimize for each quantity $Q$ of interest the prefactor $d_{+}^{(Q)}(\wb,\Eb)$ of its upper bound. In that way, we will identify the optimal vectors $\wb$ for which the predictive power of any numerical approach based on the GOK variational principle is maximized: given an error $\Delta{E}_{\wb}$ of the ensemble energy, these are exactly those $\wb$ for which the largest possible error $\Delta{Q}$ that $Q$ may have is minimal.

Before we derive the optimal weight vectors, two crucial comments are in order here concerning the practical relevance of the anticipated results and our error bounds in general. At first sight, one may wonder how our results can be applied in practice, given that they necessitate the approximate knowledge of $\DEwH$ and for some quantities $Q$ even $\Eb$. First, as far as $\DEwH$ is concerned, we recall that in numerous variational optimization methods systematic extrapolation techniques are available for determining good estimates of the exact variational energies. Exactly the same schemes could be used in the context of the GOK variational principle. Prime examples would be the schemes in the density matrix renormalization group (DMRG) ansatz for extrapolating to the results of infinite bond-dimension ~\cite{legeza1996accuracy,chan2002highly,legeza2003controlling,weichselbaum2009variational,marti2010dmrg,olivares2015dmrg,zauner2018variational,hubig2018error} and in the selected configuration interaction approach for extrapolating to the full configuration interaction ansatz involving all electron configurations~\cite{holmes2016heatbath,burton2023rationale}. Second, for those quantities $Q$ whose optimal $\wb$ will depend on the energy spectrum $\Eb$, one could get a reasonable estimate of $\mathbf{E}$ through a low-level method such as the Hartree-Fock ansatz or any cheaper post Hartree-Fock approach. Most importantly, this initial estimate of $\Eb$ and resulting $\wb$ can further be improved during the course of the variational minimization of the ensemble energy. We explain this exemplarily for the state-average CASSCF approach~\cite{mcweeny1974scf,sachs1975mcscf,deskevich2004dynamically,MatsikaRev21}. There, one first optimizes the active orbitals via a gradient step in order to minimize the ensemble energy $\Ew$. Then, in a second step, one determines the targeted eigenstates and their energies within the active space of the updated orbitals through some exact or approximate diagonalization technique. This two-step procedure is repeated numerous times until convergence is reached. Apparently, one obtains more and more accurate estimates of the energy spectrum $\Eb$ during the optimization and could therefore constantly update $\wb$ in order to maximize the predictive power of the final results.

In the following we present and discuss the optimal weight vectors which reduce the error ranges of $\deltaEk, \deltaE, \deltaPsik$ and $\deltaPsi$, respectively.
We remark that the underlying reasoning, however, would not make any sense for the error $\Delta{\rho}_{\bm{w}}$ of the ensemble state, since $\rw$ itself explicitly depends on $\wb$. In particular the upper error bound in Theorem \ref{res:ens_state_err} vanishes for $\wb \propto (1,1,\ldots,1)$ because there exists only one quantum state $\rtw= \rw = \openone/D$ for that specific weight vector. We commence with the practically most important case of targeting individual energies $E_k$. Since the error $\deltaEk$ can be also negative, we first merge  both bounds in Theorem \ref{res:ek_error} to a single one,
\begin{equation}\label{Ekboundabs}
    0 \leq |\Delta{E}_k| \leq \max\left\{\mu_{k-1}^{-1},\mu_{k}^{-1}\right\}\Delta{E}_{\bm{w}}.
\end{equation}
Here, we introduced the more instructive parameters
\begin{equation}\label{mu}
    \begin{split}
        \mu_{l} &\equiv w_{l}-w_{l+1}, \quad l < D-1
        \\
        \mu_{D-1} &\equiv w_{D-1},
    \end{split}
\end{equation}
where obviously $\mu_l \geq 0$ for all $l$ since $\wb=\wb^\downarrow$. Additionally, the normalization condition translates to
\begin{equation}\label{munorm}
    \sum_{l=0}^{D-1}w_l=\sum_{l=0}^{D-1}(l+1)\mu_l = 1,
\end{equation}
which implies in particular $\mu_l \leq 1/(l\!+\!1)$.
Throughout this section, any $\bm{\mu}$ is assumed to respect these conditions such that the corresponding $\wb$ is indeed a probability distribution. Minimizing then the prefactor of the right-hand side in \eqref{Ekboundabs} over all admissible $\mub$ yields in a straightforward manner
\begin{equation}
    \min_{\substack{\bm{\mu}:\, \mu_{k-1},\,\mu_{k}>0}} \max\left\{\mu_{k-1}^{-1},\mu_{k}^{-1}\right\} = 2k+1,
\end{equation}
realized by $\mu_{k} = \mu_{k+1} = (2k\!+\!1)^{-1}$. This translates to the optimal weight vector given by
\begin{equation}\label{eqn:optw0}
    \wb = \frac{1}{2k\!+\!1}(\underbrace{2,\ldots,2}_{k},1,\underbrace{0,\ldots,0}_{D-k-1}).
\end{equation}
In this vector, the gaps between those consecutive weights which involve the target eigenstate ($k$) are maximized, while all other gaps are minimized to 0. The corresponding lowest upper bound for $\Delta{E}_k/\Delta{E}_{\bm{w}}$ is $2k+1$, which increases linearly in $k$. Remarkably, this rationalizes in quantitative terms the common expectation that the higher an excitation lies, the larger can be the error of the corresponding energy $E_k$.

In the following we present and discuss the optimal $\wb$ and the corresponding lowest upper error bounds also for the other relevant quantities. These additional key results of our work can be derived similarly to those for $E_k$. Due to the technical character of their derivations we refer the reader to Appendix \ref{app:woptimal} for various details.

If we are interested in calculating all $D$ eigenenergies accurately, the optimal ${\wb}$ is determined by minimizing the upper bound of $\sum_{k=0}^{D-1}|\Delta{E}_k|$ in Theorem \ref{res:en_error}. By recalling the optimal $\wb$ for an individual energy $E_k$ in \eqref{eqn:optw0}, it is less surprisingly that the optimal $\wb$ follows here as
\begin{equation}
    \begin{split}
        {\wb} = \frac{2}{D(D-1)} (D\!-\!1,D\!-\!2,\ldots,1,0)\label{eqn:optw1}.
    \end{split}
\end{equation}
This means that all gaps between consecutive weights are equal and thus equally maximal. The corresponding lowest upper bound in Theorem \ref{res:en_error} follows as $D(D-1)$. If we are interested in calculating just the lowest $K$ ($0<K<D\!-\!1$) eigenenergies, which is the more sensible setting for realistic quantum systems with exceedingly large $D$,
the following ${\wb}$ turns out to be optimal
\begin{equation}
    \begin{split}
    {\wb} = \frac{1}{K^2} (2K\!-\!1,2K\!-\!3,\ldots,3,1,0,\ldots,0).
    \end{split}\label{eqn:optw2}
\end{equation}
This result makes sense as well from an intuitive point of view: for the lowest $K$ energies, the weights between consecutive energies are maximally distinct which is achieved also by assigning the weight 0 to all higher energy levels. The corresponding lowest upper bound for the ratio $\sum_{k=0}^{K-1}|\Delta{E}_k|/\Delta{E}_{\bm{w}}$ follows directly from \eqref{eqn:optw2} and Theorem \ref{res:en_error} as $K^2$, which grows quadratically as function of $K$.
%
%

To discuss the case of the $k$-th eigenstate $\ket{\Psi_k}$, we recast the upper bound in Theorem \ref{res:psik_error} as
\begin{equation}
    \Delta{\Psi}_k \leq \max\left\{\frac{r_-}{\mu_{K-1}},\frac{r_+}{\mu_{k}}\right\},
\end{equation}
where $r_{\pm} \equiv |E_{k}- E_{k\pm1}|^{-1}$. The lowest upper bound then follows as
\begin{equation}
    \min_{\bm{\mu}:\,\mu_{K-1},\,\mu_{K}>0} \max\left\{\frac{r_-}{\mu_{k-1}},\frac{r_+}{\mu_{k}}\right\} \!=\! k(r_-+r_+)+r_+,
\end{equation}
with $\mu_{k-1} = r_-/[k(r_-\!+\!r_+)\!+\!r_+]$ and $\mu_{k} = r_+/[k(r_-\!+\!r_+)\!+\!r_+]$. This leads to the optimal weight vector given by
\begin{equation}
\begin{split}
    \bm{w} = &\frac{1}{k(r_-\!+\!r_+)\!+\!r_+} \times
    \\
    &\quad (\underbrace{r_-\!+\!r_+,\ldots,r_-\!+\!r_+}_{k}, r_+,\underbrace{0,\ldots,0}_{D-k-1}).
    \end{split}
\end{equation}
In contrast to the optimal weights for $\deltaEk$ and $\deltaE$, which are independent of the energy spectrum $\Eb$, the optimal gaps $\mu_j$ between consecutive weights are now balanced by the corresponding energy gaps. It are not the gaps that are maximized by the optimal weight vector but the product of the weight and energy gaps. In a similar fashion
we find also the optimal weights for targeting multiple eigenstates. When all $D$ eigenstates are considered, the optimal $\bm{w}$ for achieving the lowest upper bound for $\sum_{k=0}^{D-1}\Delta\Psi_k$ follows as
\begin{equation}
    \bm{w} \propto \left(\sum_{j=k+1}^{D-1}(E_j-E_{j-1})^{-1}\right)_{k=0}^{D-1},
\end{equation}
whereas when only the lowest $K\!<\!D\!-\!1$ eigenstates are considered, the optimal $\bm{w}$ for achieving the lowest upper bound for $\sum_{k=0}^{K-1}\Delta\Psi_k$  follows as
\begin{equation}
    \bm{w} \propto \left(\sum_{j=k+1}^{D-1}(E_j-E_{j-1})^{-1}\Theta(K-j)\right)_{k=0}^{D-1}.
\end{equation}

We summarize the optimal weight vectors and lowest upper bounds for all relevant quantities in Table \ref{tab:opt_bounds}. It is worth recalling here that this second key accomplishment of our work was only possible as we succeeded in the first place in deriving the universal and optimal error bounds presented in Theorems \ref{res:ens_state_err}\!-\!\ref{res:en_error}. With the `instruction manual' Table \ref{tab:opt_bounds} at hand, the optimal $\wb$ can be chosen for any quantity $Q$ of interest. This will then maximize the accuracy of any numerical approach based on the GOK variational principle such as GOK-DFT~\cite{PhysRevA.37.2809,Carsten2017,PhysRevB.95.035120,Fromager18,loos2020weight,Fromager20, Fromager21,Gross21,gould2021ensemble,PhysRevA.104.052806,lu2022multistate,Loos23JPCLett,PhysRevB.109.235113}, $\wb$-RDMFT~\cite{PhysRevLett.127.023001,liebert2021foundation,JLthesis,LS23SciP,LS23NJP}, specific quantum Monte Carlo techniques~\cite{schautz2004optimized,schautz2004excitations,Filippi}, traditional quantum chemical methods~\cite{mcweeny1974scf,sachs1975mcscf,deskevich2004dynamically,PhysRevLett.129.066401,MatsikaRev21} and quantum computing~\cite{ssvqe,Higgott2019variationalquantum,Cerezo2022,Nakanishi,Yalouz_2021,Yalouz22,PhysRevA.107.052423,hong2023quantum,benavides2023quantum}.

\section{Illustration: Variational Quantum Eigensolver}\label{sec:VQE}
To showcase the practical relevance of our results, we apply our theorems within the framework of variational quantum eigensolvers (VQE) ~\cite{Peruzzo2014}, used to solve the transverse Ising model
\begin{equation}
    \hat{H}= \sum_{i=1}^N a_i \hat{X_i} + \sum_{1 \leq i < j \leq N}   J_{ij} \hat{Z_i}\hat{Z_j}, \label{eqn:ham_qc}
\end{equation}
where $\hat{X_i}, \hat{Z_i}$ denote the Pauli operators for the $i$-th spin. For this proof of concept, it is sufficient to chose $N=2$, resulting in a Hilbert space of dimension $D=2^2$. The coefficients $a_i$ and $J_{ij}$ are fixed random numbers uniformly sampled from the interval $[0,1)$. In our work they are given by $J_{12}=0.09000 , a_1= 0.32696, a_2=0.80430$ which results in the following eigenenergies
\begin{equation}
    \begin{split}
        E_0 &= -1.13483,
        \\
        E_1 &= -0.48575,
        \\
        E_2 &= \phantom{-}0.48575,
        \\
        E_3 &= \phantom{-}1.13483.
    \end{split}
\end{equation}

The quantum simulations were conducted using Pennylane ~\cite{pennylane}, a quantum machine learning and simulation library. To diagonalize the Hamiltonian \eqref{eqn:ham_qc} in virtue of the ensemble variational principle through a unitary transformation $\hat{U}$, we resort to the parametric ansatz proposed in Ref.~~\cite{QC-ansatz}. By increasing the circuit depth, the ansatz for $\hat{U}$ is then capable of diagonalizing the Hamiltonian \eqref{eqn:ham_qc} with arbitrary precision. Three different weight vectors $\wb$ are employed,
\begin{equation}\label{wVQE}
    \bm{w}^{(n)} \propto (4^n,3^n,2^n,1^n), \quad n=1,2,3
\end{equation}
to demonstrate the effect of the choice of $\wb$ in the error of the quantum states and eigenenergies during the optimization of the circuit parameters using the Adam algorithm~\cite{kingma2017adam}.

In Figure \ref{fig:qc-vqe}, with present the errors of the ensemble state ($\Delta\rho_{\wb}$), the first excited state ($\Delta \Psi_1$) and eigenenergy ($\Delta E_1$), as well as the accumulated errors of various eigenstates ($\deltaPsi\equiv \sum_{k=0}^{3}\Delta\Psi_k$) and eigenenergies ($\deltaE\equiv\sum_{k=0}^3|\Delta{E}_k|$), against the error $\Delta{E}_{\wb}$ of the ensemble energy.
The red data points record the evolution of the errors during the optimization of the circuit parameters, while the upper and lower bounds of the errors given by Theorems \ref{res:ens_state_err}~-~\ref{res:en_error} are shown as black lines. In contrast to Figure \ref{fig:sampling1} and \ref{fig:sampling2}, we opted here for a double-logarithmic scale for various error types except for $\Delta{E}_1$ (which can be negative), in order to resolve the behavior of the errors in the important regime where $\Delta{E}_{\wb}$ is arbitrarily small (regime of convergence).

By following in each plot from the right side the red data points we observe that during the course of the optimization the variational ensemble energy $\Ew$ gets more and more accurate.
Since the individual and accumulated errors $\Delta Q$ of the eigenenergies and eigenstates are linearly bounded from above by $\DEwH$ according to Theorems \ref{res:ens_state_err}~-~\ref{res:en_error}, the excellent convergence of the latter enforces an equally excellent (or even better) convergence of the former. In particular, for the error quantities which are bounded from above and below by positive bounds (first, third and fifth row), the corresponding black lines form ``channels'' which ``guides'' the respective errors $\Delta Q$ to the regime of excellence convergence. Note here, that the linear positive bounds $d_-\Delta\Ew \leq  \Delta Q \leq d_+ \Delta\Ew$ in Theorems \ref{res:ens_state_err}, \ref{res:eigenstate_error}, \ref{res:en_error} translate here to
\begin{equation}
\begin{split}
 \log (d_-) + \log(\Delta\Ew) &\leq  \log(\Delta Q) 
 \\
 &\leq \log(d_+) + \log(\Delta\Ew).
 \end{split}
\end{equation}
Accordingly, the vertical arrangement of any black line is defined by an offset $\log(d_{\pm})$ and the width of each ``channel'' follows as $\log(d_+)-\log(d_-)$.
The other two error quantities (second and fourth row) are also bounded from above by the ensemble error but they do not have a positive lower bound. This means that they may vanish even if the ensemble energy has an error $\DEwH>0$. Indeed, this occasionally happens during the course of the optimization for the error $\Delta E_1$ which changes signs from time to time.

\begin{figure*}[t]
    \centering
    \includegraphics[width=0.8\linewidth]{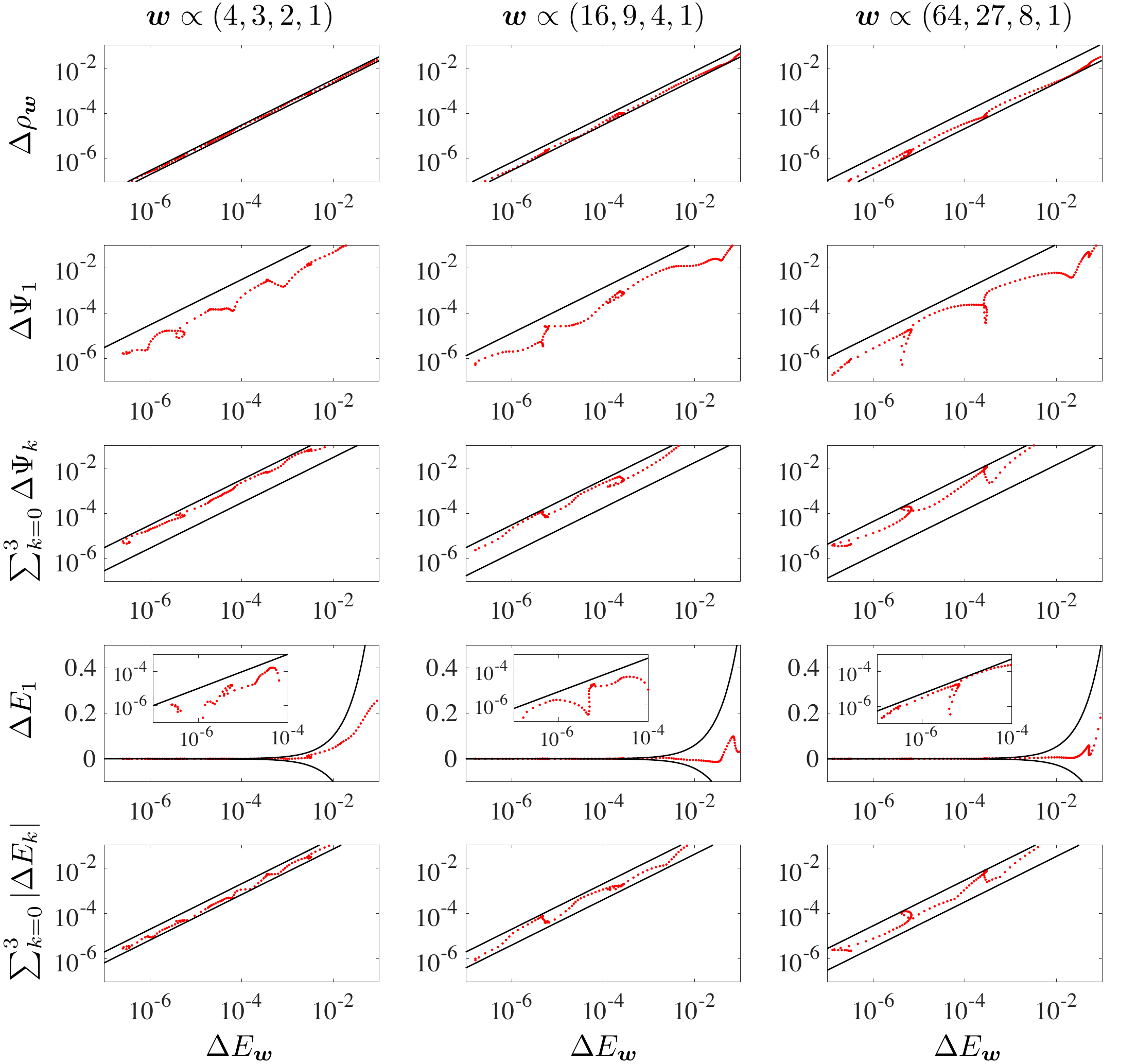}
    \caption{Illustration of the practical relevance of our main results, Theorems \ref{res:ens_state_err}~-~\ref{res:en_error}, in the context of a variational quantum eigensolver applied to the transverse field Ising model: During the course of optimization (results depicted as red points), our theoretical bounds shown as black solid lines guarantee convergence of various relevant quantities as the ensemble error $\DEwH$ gets arbitrarily small.}
    \label{fig:qc-vqe}
\end{figure*}

Comparison of the three columns confirms again that the choice of the weights can have a strong impact on the narrowness of the error ``channels'' of $\deltarhow, \deltaPsi, \deltaE$ and the vertical arrangements of various black lines. For instance, according to the bounds in Theorem \ref{res:ens_state_err}, the ``channel'' of $\deltarhow$ could become arbitrarily narrow by making various weights more and more equal. Even changing the weights from those of the last two columns to just an equally spaced $\wb$ (first column) has already a huge effect on the narrowness. Yet, in that context one must not forget that changing $\wb$ even further to a constant weight vector $\wb \propto (1,\ldots,1)$ would trivialise the search space underlying the GOK variational principle (as explained above) and thus the latter would loose its  predictive power entirely. Moreover, we observe in agreement with the predictions made by Table \ref{tab:opt_bounds} that the first weight vector, which is close to the optimal one for the accumulated energy error $\Delta E$, lowers indeed the upper bound compared to the other two $\wb$, yet the effect is not that drastic. At the same time, this equally spaced $\wb$ is far from being the optimal choice for the other error quantities (c.f Table \ref{tab:opt_bounds}). For instance, in the fourth row a wider error range follows compared to the second and third column. That is not surprising since the tightest bounds on $\Delta E_1$ are those obtained by maximizing either $w_1-w_2$ (upper bound) or $w_0-w_1$ (lower bound) with the effect of sacrificing  the accuracy of all other eigenstates and eigenenergies. This emphasizes the conceptual difference between the errors of individual excited states and the error of the ensemble state. This in turn suggests a potential for error cancellation due to the absence of a variational principle for individual excited states. All these effects will be even more well-pronounced in bigger systems with larger Hilbert spaces.

Finally, we observe that many of the error bounds in our example are occasionally (approximately) saturated during the optimization of the circuit parameters. Such a saturation actually occurs for all three weight vectors. This further demonstrates the particular relevance of our analytic error bounds which apparently can dictate the behavior of the errors $\Delta Q$ as function of $\Delta \Ew$. The worst-case scenario where the error of the eigenstates or eigenenergies is maximized for a fixed $\Delta{E}_{\wb}$ could indeed occur during the optimization process. From this concrete example, we therefore conclude that the weight vector $\wb$ of the ensemble state is more than an auxiliary quantity, and that a carefully chosen $\wb$ can lead to drastic improvements in the errors of the individual eigenstates and eigenenergies, and significantly improve the predictive power of the GOK variational principle.

\section{Summary and Conclusions}\label{sec:concl}
We have provided a sound basis for the ensemble variational principle~\cite{GOK1988} which has been proposed by Gross, Oliveira and Kohn as a seminal extension of the famous Rayleigh-Ritz variational principle. It minimizes the energy expectation value $\mathrm{Tr}[\rtw \hat{H}]$
of a quantum system over ensemble states $\rtw \equiv \sum_{k=0}^{K-1} w_k |\tilde{\Psi}_k\rangle\langle\tilde{\Psi}_k|$ with fixed auxiliary spectrum $\wb$ in order to target excited states by variational means. By employing elegant concepts from convex geometry we confirmed by constructive means the predictive power of this GOK variational principle and proved the validity of the underlying critical hypothesis: Whenever the ensemble energy is well-converged, the same holds true for the individual eigenstates  and eigenenergies $E_k$, provided various auxiliary weights $w_k$ were chosen differently. To be more specific, as our main accomplishment, we derived linear bounds
\begin{equation}
d_{-}^{(Q)}(\wb,\Eb)\,\Delta\Ew \leq  \Delta Q \leq d_{+}^{(Q)}(\wb,\Eb)\,\Delta\Ew
\end{equation}
on the errors $\Delta Q $ of these sought-after quantities. A detailed analytical and numerical analysis in Section \ref{sec:saturation} has confirmed the tightness and practical significance of these error bounds. For instance, in practical applications such as the variational quantum eigensolver (Section \ref{sec:VQE}), they dictate the convergence rate of various physical quantities as function of the error of the variational ensemble energy $\Ew$.

Moreover, a comprehensive analysis of the explicit form of the prefactors $d_{\pm}^{(Q)}(\wb,\Eb)$ allowed us for each physical quantity of interest to determine the corresponding optimal auxiliary weights $\wb$. The corresponding Table \ref{tab:opt_bounds} which summarizes these key results will serve as an `instruction manual' for practical applications. Following it will then maximize the predictive power of the GOK variational principle, and thus the accuracy of any ensemble-based numerical approach, such as GOK-DFT~\cite{PhysRevA.37.2809,Carsten2017,PhysRevB.95.035120,Fromager18,loos2020weight,Fromager20, Fromager21,Gross21,gould2021ensemble,PhysRevA.104.052806,lu2022multistate,Loos23JPCLett,PhysRevB.109.235113}, $\wb$-RDMFT~\cite{PhysRevLett.127.023001,liebert2021foundation,JLthesis,LS23SciP,LS23NJP}, specific quantum Monte Carlo techniques~\cite{schautz2004optimized,schautz2004excitations,Filippi}, traditional quantum chemical methods~\cite{mcweeny1974scf,sachs1975mcscf,deskevich2004dynamically,PhysRevLett.129.066401,MatsikaRev21} and quantum computing~\cite{ssvqe,Higgott2019variationalquantum,Cerezo2022,Nakanishi,Yalouz_2021,Yalouz22,PhysRevA.107.052423,hong2023quantum,benavides2023quantum}.

In summary, our work cemented the GOK variational principle as the cornerstone for excited state methods in complete analogy to the Rayleigh-Ritz variational principle for ground state methods. From a practical point of view, our analytical error bounds reveal the fundamental guiding principle for designing optimal weight vectors $\wb$ for computing excited states and energies in physics, chemistry and materials science.

\begin{acknowledgments}
We are grateful to C.L.~Benavides-Riveros, K.~Liao, J.~Liebert, S.~Paeckel and Z.~Xie for inspiring discussions.
We acknowledge financial support from the German Research Foundation (Grant SCHI 1476/1-1), the Munich Center for Quantum Science and Technology, and the Munich Quantum Valley, which is supported by the Bavarian state government with funds from the Hightech Agenda Bayern Plus.
\end{acknowledgments}

\bibliography{gok_quantum}

\onecolumn\newpage

\appendix

\section{Proofs of Theorems}

In this section we provide the proofs to the theorems in the main text.
Most of the proofs will resort to the geometric picture of the Birkhoff polytope or permutohedron.
To be more specific, we will formulate the proofs in terms of a linear constrained optimization problem on a convex polytope.
The linear constraint is represented by a hyperplane $\mathbb{A}(\delta)$ (defined by the equality constraint $\Delta{E}_{\wb} = \delta$) that cuts through the polytope and intersects at various edges.
Each intersected edges connects two vertices on the polytope,
one on the negative side of $\mathbb{A}(\delta)$ ($\Delta{E}_{\wb} < \delta$), which we will call a $(-)$-vertex (in the main text this is called a reference vertex),
and the other on the positive side ($\Delta{E}_{\wb} > \delta$), which we will call a $(+)$-vertex.
By identifying all $(\pm)$-vertices one effectively identifies all edges intersected by the linear constraint.
The linear constrained convex optimization problem is then reduced to the convex optimization problem within the convex intersection between the constraint and the polytope.

\subsection{Proof of Theorem \ref{res:ens_state_err}}
\begin{proof}
    We first reformulate the statement in terms of the permutohedron $P({\wb})$.
    Recall that
    \begin{eqnarray}
        \Delta\rho_{\wb}(\tilde{{\wb}}) &=& 2{\wb}^\mathrm{T}(\wb-\tilde{{\wb}}),
        \\
        \Delta{E}_{\wb}(\tilde{{\wb}}) &=& (\tilde{{\wb}}-\wb)^\mathrm{T}\mathbf{E}, \label{app:eqn:ens_error}
    \end{eqnarray}
    where $\tilde{{\wb}} \in P({\wb})$.
    We wish to determine the range of $\Delta\rho_{\wb}(\tilde{{\wb}})$,
    for a given value of $\Delta{E}_{\wb}(\tilde{{\wb}}) = \delta$ where $0 < \delta < g_{\bm{w},\mathbf{E}}$.

    In that case, one easily verifies that the only $(-)$-vertex is ${\wb}$.
    The type of neighboring vertices to $\wb$ (the $(+)$-vertices) depends on the number of unique elements in $\wb$.
    If $\wb$ is such that $w_0 > w_1 > \cdots > w_{D-1}$, then there are $D\!-\!1$ neighboring vertices, namely the $D\!-1\!$ adjacent transpositions of $\wb$.
    If $\wb$ is such that $w_0 > w_1 > \cdots > w_{K} = \cdots = w_{D-1}=0$ for some $K<D\!-\!1$, there are also $D\!-\!1$ neighboring vertices, but they consist of $K$ adjacent transpositions $\mathbf{S}_{k,k+1}$ where $0\leq k < K$, and $D\!-\!K$ non-adjacent transpositions $\mathbf{S}_{K-1,k+1}$ where $k\geq K$. The index tuples $(k,l)$ such that $\mathbf{S}_{k,l}\wb$ is a $(+)$-vertex is conveniently grouped by the set $\mathcal{I}_{\bm{w}}$.

    The linear plane $\mathbb{A}(\delta)$ defined by $\Delta{E}_{\bm{w}}(\tilde{\bm{w}}) = \delta$ intersects all edges connecting the vertex ${\wb}$, at points
    \begin{equation}
        \mathbf{v}_{k,l} = (1-p_{k,l}) {\wb} + p_{k,l} \mathbf{S}_{k,l}{\wb},\quad (k,l)\in\mathcal{I}_{\wb}
    \end{equation}
    where $p_{k,l} = \delta/[(w_k\!-\!w_{l})(E_{l}\!-\!E_{k})]$.
    Since the intersection $\mathbb{A}(\delta)\cap P({\wb})$ is again convex with extremal points being the $\mathbf{v}_{k,l}$'s with $(k,l)\in\mathcal{I}_{\bm{w}}$,
    we can deduce that for any $\mathbf{v} \in \mathbb{A}\cap P({\wb})$
    \begin{equation}
        \begin{split}
            \min_{(k,l)\in\mathcal{I}_{\wb}} \Delta\rho_{\wb}(\mathbf{v}_{k,l}) \leq \Delta\rho_{\wb}(\mathbf{v}) \leq \max_{(k,l)\in\mathcal{I}_{\wb}} \Delta\rho_{\wb}(\mathbf{v}_{k,l}).
        \end{split}
    \end{equation}
    By substituting in the exact expression of the error of the ensemble state at the points $\mathbf{v}_{k,l}$'s
    \begin{equation}
        \Delta\rho_{\wb}(\mathbf{v}_{k,l}) = 2\delta\frac{w_k-w_{l}}{E_{l}-E_{k}},
    \end{equation}
    we arrive at both the lower and upper bounds.
\end{proof}

\subsection{Proof of Theorem \ref{res:psik_error}}

\begin{proof}
    Recall that
    \begin{eqnarray}
        \Delta\Psi_k(\mathbf{X}) &=& 1-X_{kk},
        \\
        \Delta{E}_{\wb}(\mathbf{X}) &=& {\wb}^\mathrm{T}(\mathbf{X}^\mathrm{T}-\openone)\mathbf{E},
    \end{eqnarray}
    where $X_{kl} = |\langle\Psi_k|\tilde{\Psi}_l\rangle|^2$ is an unistochastic matrix in the Birkhoff polytope $B_D$.
    We are interested in the region where $0 < \Delta{E}_{\wb} = \delta < g_{{\wb},\mathbf{E}}$.

    It is not difficult to see that the lower bound for $\Delta\Psi_k$ is zero,
    since a finite error in the ensemble energy can be due to errors in other eigenenergies than $E_k$.
    For the upper bound, let us suppose ${\wb}$ contains $K$ positive entries.
    Then the $(-)$-vertices on the Birkhoff polytope are $\openone$ and permutations
    $\mathbf{S}$ that act only non-trivially on the zeros elements of ${\wb}$.
    According to the property of the Birkhoff polytope\cite{brualdi1977convex}, $(+)$-vertices are permutations that differ by a single cycle from the $(-)$-vertices.
    Notice that for any permutation $\mathbf{S}$, the resulting error $\Delta\Psi_k$ in the $k$-th eigenstate can only be 0 or 1.
    Therefore, to find the upper bound of $\Delta\Psi_k$,
    it suffices to find among the $(+)$-vertices involving the $k$-th eigenstate,
    that realizes the smallest error in the ensemble energy (and thus maximizing the ratio $\Delta\Psi_k/\Delta{E}_{\wb}$).
    Such smallest error is given by $\min\{g_{k-1,k}({\wb},\mathbf{E}),g_{k,k+1}({\wb},\mathbf{E})\}$, for $k\geq1$, where $g_{kl}(\bm{w},\mathbf{E}) \equiv (w_{k}\!-\!w_{l})(E_l\!-\!E_k)$.
    When $k=0$, the smallest error is simply $g_{0,1}({\wb},\mathbf{E})$.
    We therefore conclude that
    \begin{equation}
        0 \!\leq\! \frac{\Delta\Psi_k}{\Delta{E}_{\bm{w}}} \!\leq\! \begin{cases}
            g^{-1}_{0,1}({\wb},\mathbf{E}),  &k=0
            \\
            \max\{g^{-1}_{k-1,k}({\wb},\mathbf{E}),g^{-1}_{k,k+1}({\wb},\mathbf{E})\},  &k \geq 1.
        \end{cases}
    \end{equation}
\end{proof}

\subsection{Proof of Theorem \ref{res:eigenstate_error}}

\begin{proof}
    The proof follows similarly to the proof of Theorem \ref{res:psik_error}. We assume there are exactly $K$ positive elements in $\bm{w}$.
    The $(-)$-vertices are still the identity matrix $\openone$ and permutations $\mathbf{S}_-$ that acts only non-trivially on the zero elements of ${\wb}$.
    And the $(+)$-vertices are permutations that differ by a single cycle from the $(-)$-vertices.

    A key difference is that for a $(+)$-vertex $\mathbf{S}_+$, the corresponding error of the $K$ lowest eigenstates is
    \begin{equation}
        1 \leq \sum_{k=0}^{K-1}\Delta\Psi_k(\mathbf{S}_+) \leq K.
    \end{equation}
    Notice that the error cannot be zero, as that can only be realized by one of the $(-)$-vertices.
    When ${\wb}$ is strictly decreasing ($K=D$), the lower bound is improved to 2, since the minimal number of misplaced eigenstates is 2 in that case.
    The exact (integer) value of the error depends on the number of non-zero elements of ${\wb}$ that $\mathbf{S}_+$ non-trivially acts on. We define the following two helper quantities
    \begin{equation}
        \begin{split}
            \delta^{(1)}_{\bm{w},\mathbf{E}} &= \min_{i,j < K} (w^\downarrow_i-w^\downarrow_j)(E^\uparrow_j-E^\uparrow_i),
            \\
            \delta^{(2)}_{\bm{w},\mathbf{E}} &= \min_{i< K,j\geq K} (w^\downarrow_i-w^\downarrow_j)(E^\uparrow_j-E^\uparrow_i).
        \end{split}
    \end{equation}

    \begin{lemm} \label{lemma}
        Let $\mathbf{S}_L=\mathbf{S}_{\mathbf{n}}\mathbf{S}_{-}$ be a $(+)$-vertex of the Birkhoff polytope with respect to the hyperplane $\mathbb{A}$ defined by $0< \Delta{E}_{\wb}=\delta<g_{{\wb},\mathbf{E}}$.
        Here, $\mathbf{S}_-$ is a $(-)$-vertex and $\mathbf{S}_{\mathbf{n}}$ is a single $L$-cycle ($L\geq 2$) denoted in cycle notation as $\mathbf{n}=(n_1,n_2,\cdots,n_L)$
        where exactly $L$ many $w_{n_k}$'s are positive and distinct ($1\leq L' \leq L$), while the others are zero. Then we have
        \begin{equation}
        \begin{split}
            \Delta{E}_{\wb}(\mathbf{S}_L) \leq  \min\left(L', \left\lfloor \frac{L}{2} \right\rfloor \right) G_{{\wb},\mathbf{E}}.
        \end{split}
        \end{equation}
        When $\delta^{(1)}_{{\wb},\mathbf{E}} \leq 2 \delta^{(2)}_{{\wb},\mathbf{E}}$ and $L'>1$, we have
        \begin{equation}
            (L'-1)\delta^{(1)}_{{\wb},\mathbf{E}} \leq \Delta{E}_{\wb}(\mathbf{S}_L),
        \end{equation}
        and when $\delta^{(1)}_{{\wb},\mathbf{E}} > 2 \delta^{(2)}_{{\wb},\mathbf{E}}$ or $L'=1$, we have
        \begin{equation}
            (2L'-1)\delta^{(2)}_{{\wb},\mathbf{E}} \leq \Delta{E}_{\wb}(\mathbf{S}_L).
        \end{equation}
        Specially, when $L'=L$, the inequalities become
        \begin{equation}
            \begin{split}
                (L-1)g_{{\wb},\mathbf{E}} \leq \Delta{E}_{\wb}(\mathbf{S}_L) \leq \left\lfloor\frac{L}{2}\right\rfloor G_{{\wb},\mathbf{E}}.
            \end{split}
        \end{equation}
    \end{lemm}
    \begin{proof}
    For the upper bound, we use the rearrangement inequality
    \begin{equation}
        \begin{split}
            \Delta{E}_{\wb}(\mathbf{S}_L) & \leq \sum_{l=1}^{\left\lfloor L/2 \right\rfloor} (w_{n^\uparrow_{l}}\!-\!w_{n^\uparrow_{L-l+1}})(E_{n^\uparrow_{L-l+1}}\!-\!E_{n^\uparrow_l})
            \\
            &\leq \min\left(L', \left\lfloor \frac{L}{2} \right\rfloor \right) G_{{\wb},\mathbf{E}}.
        \end{split}
    \end{equation}
    The lower bound can be proven by induction. When $L=2$, the lower bounds for $\Delta{E}_{\bm{w}}(\mathbf{S}_L)$ are trivially satisfied.
    Let $\mathbf{S}_{L+1}=\mathbf{S}_{\mathbf{n}}\mathbf{S}_-$
    where $\mathbf{S}_{\mathbf{n}}$ is a $(L+1)$-cycle with
    $\mathbf{n}=(n_1,n_2,\ldots,n_K,n_{L+1})$ and $L\geq 2$.
    Using the cyclic properties of the cycle notation,
    we can assume that $n_{L+1}$ is the index corresponding to the largest energy (or equivalently the smallest weight).
    By the assumption of the induction, the $L$-cycle $\mathbf{S}_L$ representing $(n_1,n_2,\ldots,n_{L})$ satisfies the lower bound. Then
    \begin{equation}
        \begin{split}
            &\quad \:{\wb}^{\mathrm{T}}(\mathbf{S}_{L+1}\!-\!\openone)\mathbf{E}
            \\
            &= \sum_{l=1}^{L-1}w_{n_l}(E_{n_{l+1}}\!-\!E_{n_l})+ w_{n_L}(E_{n_{L+1}}\!-\!E_{n_L})
            \\
            &\quad+ w_{n_{L+1}}(E_{n_1}\!-\!E_{n_{L+1}})
            \\
            &= {\wb}^{\mathrm{T}}(\mathbf{S}_{L}\!-\!\openone)\mathbf{E} + (w_{n_L}\!-\!w_{n_{L+1}})(E_{n_{L+1}}\!-\!E_{n_1}).
        \end{split}
    \end{equation}
    Each step of induction acquires an additional term.
    Notice that the additional term vanishes if and only if both $w_{n_L}$ and $w_{n_{L+1}}$ are zero.
    When at least one of them is positive, we have the following inequality
    \begin{equation}
        \min(\delta^{(1)}_{{\wb},\mathbf{E}},\delta^{(2)}_{{\wb},\mathbf{E}}) \leq (w_{n_L}\!-\!w_{n_{L+1}})(E_{n_{L+1}}\!-\!E_{n_1}).
    \end{equation}
    Depending on the order of the additional index $n_{L+1}$ at each step of the induction, which we call a line-up, the above bound can be made tighter. Namely, if two consecutive elements $n_{L}$ and $n_{L+1}$ are positive, we can conclude the lower bound of the additional term is $\delta^{(1)}_{\bm{w},\mathbf{E}}$. If only one of them is positive, the lower bound becomes $\delta^{(1)}_{\bm{w},\mathbf{E}}$.
    If we use $1$ to represent a positive $w_{n_l}$ and $0$ for a vanishing one,
    then the line-up that minimizes the accumulated sum in the induction is one of the following two cases.
    When $\delta^{(1)}_{{\wb},\mathbf{E}} \geq 2 \delta^{(2)}_{{\wb},\mathbf{E}}$ and $L'>1$, the minimizing line-up is
    $$
    \underbrace{1,1,\ldots,1}_{L'},0,0,\ldots,0.
    $$
    The accumulated sum in the induction is then
    \begin{equation}
        \begin{split}
            (L'-1)\delta^{(1)}_{{\wb},\mathbf{E}} \leq \Delta{E}_{\wb}(\mathbf{S}_L).
        \end{split}
    \end{equation}
    When $\delta^{(1)}_{{\wb},\mathbf{E}} > 2 \delta^{(2)}_{{\wb},\mathbf{E}}$ or $L'=1$, the minimizing line-up is
    \begin{eqnarray*}
        &\underbrace{1,0,1,0,\ldots,1,0}_{2L'},0,\ldots,0, \quad &L' \leq L/2,
        \\
        &\underbrace{0,1,0,1,\ldots,0,1}_{2(L-L')},1,\ldots,1, \quad &L' > L/2.
    \end{eqnarray*}
    Straightforward algebra shows that the accumulated sum in the induction is then
    \begin{equation}
        (2L'-1)\delta^{(2)}_{{\wb},\mathbf{E}} \leq \Delta{E}_{\wb}(\mathbf{S}_L).
    \end{equation}
    \end{proof}

    From Lemma \ref{lemma}, we can derive that for a $(+)$-vertex $\mathbf{S}_L$ acting non-trivially on $L'$ non-zeros elements of ${\wb}$,
    the following lower bound on the error of the lowest $K$ ($K\!<\!D\!-\!1$) eigenstates
    \begin{equation}
        \frac{\sum_{k=0}^{K-1}\Delta\Psi_k(\mathbf{S}_L)}{\Delta{E}_{\wb}(\mathbf{S}_L)} \geq \frac{L'}{ \min\left(L', \left\lfloor \frac{L}{2} \right\rfloor \right) G_{{\wb},\mathbf{E}}} \geq \frac{1}{G_{{\wb},\mathbf{E}}}.
    \end{equation}
    When $K=D$, we always have $L'=L$, which gives us a tighter lower bound
    \begin{equation}
        \frac{\sum_{k=0}^{D-1}\Delta\Psi_k(\mathbf{S}_L)}{\Delta{E}_{\wb}(\mathbf{S}_L)} \geq \frac{2}{G_{{\wb},\mathbf{E}}}.
    \end{equation}

    For the upper bound, we again separate the two cases. According to Lemma \ref{lemma}, when $\delta^{(1)}_{{\wb},\mathbf{E}} \leq 2 \delta^{(2)}_{{\wb},\mathbf{E}}$ and $L'>1$
    \begin{equation}
        \frac{\sum_{k=0}^{K-1}\Delta\Psi_k(\mathbf{S}_L)}{\Delta{E}_{\wb}(\mathbf{S}_L)} \leq \frac{L'}{(L'-1)\delta^{(1)}_{{\wb},\mathbf{E}}} \leq \frac{2}{\delta^{(1)}_{{\wb},\mathbf{E}}},
    \end{equation}
    where equality is achieved when $L'=2$.
    When $\delta^{(1)}_{{\wb},\mathbf{E}} > 2 \delta^{(2)}_{{\wb},\mathbf{E}}$ or $L'=1$,
    \begin{equation}
        \begin{split}
        \frac{\sum_{k=0}^{K-1}\Delta\Psi_k(\mathbf{S}_L)}{\Delta{E}_{\wb}(\mathbf{S}_L)} &\leq \frac{L'}{(2L'-1)\delta^{(2)}_{{\wb},\mathbf{E}}}\leq \frac{1}{\delta^{(2)}_{{\wb},\mathbf{E}}},
        \end{split}
    \end{equation}
    where equality is achieved when $L'=1$. Specially, when $K=D$, we always have $L'=L\geq2$. Therefore
    \begin{equation}
        \frac{\sum_{k=0}^{D-1}\Delta\Psi_k(\mathbf{S}_L)}{\Delta{E}_{\wb}(\mathbf{S}_L)} \frac{L}{(L-1)g_{\bm{w},\mathbf{E}}} \leq \frac{2}{g_{\bm{w},\mathbf{E}}}.
    \end{equation}
\end{proof}

\subsection{Proof of Theorem \ref{res:en_error}}

\begin{proof}
    Let us first rewrite the ensemble energy error as
    \begin{equation}
    \begin{split}
        \Delta{E}_{{\wb}} &= \sum_{k=0}^{D-1} (w_k\!-\!w_{k+1}) \sum_{j=0}^{k}(\tilde{E}_j \!-\! E_j)
        \\
        &= \sum_{k=0}^{D-2} (w_k\!-\!w_{k+1}) X_k,
    \end{split}
    \end{equation}
    where $w_{D} \equiv 0$ and $F_k = \sum_{j=0}^k(\tilde{E}_j-E_j)$.
    $F_k$ are non-negative thanks to Theorem \ref{thrm:kyfan}, and specially $F_{D-1} = \mathrm{Tr}[\hat{H}-\hat{H}]=0$.
    If ${\wb}$ is strictly decreasing, then (defining $F_{-1}=0$)
    \begin{equation}
        \begin{split}
            \sum_{k=0}^{D-1}|\Delta{E}_k| &= \sum_{k=0}^{D-1}|F_k\!-\!F_{k-1}| \leq 2\sum_{k=0}^{D-1}F_k
            \\
            &\leq \frac{2}{\min_{k<D-1}(w_k\!-\!w_{k+1})}\sum_{k=0}^{D-2}(w_k\!-\!w_{k+1})F_k
            \\
            &= \frac{2\Delta{E}_{{\wb}}}{\min_{k<D-1}(w_k\!-\!w_{k+1})}.
        \end{split}
        \end{equation}
    For the lower bound, we first define the positive and negative parts (both as positive values) $\Delta\mathbf{\tilde{E}}_{\pm}$ of the vector $\Delta\mathbf{\tilde{E}}=(\Delta{E}_k)_{k=0}^{D-1}$,
    where the negative and positive entries are set to zero, respectively.
    \begin{equation}
        \begin{split}
            \Delta{E}_{{\wb}} &= \sum_{k=0}^{D-1} w_k [(\Delta{E}_+)_k - (\Delta{E}_-)_k]
            \\
            &\leq w_0 \sum_{k=0}^{D-1} (\Delta{E}_+)_k - w_{D-1} \sum_{k=0}^{D-1} (\Delta{E}_+)_k
            \\
            &= \frac{w_0\!-\!w_{D-1}}{2}\sum_{k=0}^{D-1}|\Delta{E}_k|.
        \end{split}
    \end{equation}
    In the last line we used the identity (recall that $\sum_{k=0}^{D-1}(\Delta{E}_+)_k - \sum_{k=0}^{D-1}(\Delta{E}_-)_k = F_{D-1} = 0$)
    \begin{equation}
        \sum_{k=0}^{D-1}(\Delta{E}_+)_k = \frac{1}{2}\sum_{k=0}^{D-1}|\Delta{E}_k|.
    \end{equation}

    If $w_{K-1} > w_K > w_{K+1}$ for some $0<K<D-1$,
    then the total error of the lowest $K$ eigenenergies are bounded from above similarly as
    \begin{equation}
    \begin{split}
        \sum_{k=0}^{K-1}|\Delta{E}_k| &= \sum_{k=0}^{K-1}|F_k\!-\!F_{k-1}|\leq 2\sum_{k=0}^{K-2}F_k \!+\! F_{K-1}
        \\
        &\leq \Delta{E}_{{\wb}} \max_{k<K}\left\{\frac{2\Theta(K\!-\!k\!-\!1)}{w_k\!-\!w_{k+1}}\right\}.
    \end{split}
    \end{equation}
    For the lower bound, we simply observe that
    \begin{equation}
        \tilde{E}_{{\wb}} \leq \sum_{k=0}^{K-1} w_k |\Delta{E}_k| \leq w_0 \sum_{k=0}^{K-1} |\Delta{E}_k|,
    \end{equation}
    which completes the proof.
\end{proof}

\section{Derivation for optimal weight vectors}\label{app:woptimal}
Recall the $\bm{\mu}$-notation for parameterizing a non-increasing weight vector $\bm{w}$
\begin{equation}
    \begin{split}
        \mu_k &= w_k-w_{k+1}, \quad  k < D-1.
        \\
        \mu_{D-1} &= w_{D-1}
    \end{split}
\end{equation}
where $\mu_k$'s satisfy
\begin{equation}
    \begin{split}
    0 \leq \mu_k \leq \frac{1}{k+1},
    \\
    \sum_{k=0}^{D-1} (k+1)\mu_k = 1.
    \end{split}
\end{equation}
Now we would like to solve the following minimization for the linear upper bound of the error of all eigenenergies $\sum_{k=0}^{D-1}|\Delta{E}_k|$
\begin{equation}
    \min_{\bm{\mu}} \max_{k<D-1} \frac{1}{\mu_k}.
\end{equation}
Notice that $\mu_{D-1}$ does not enter the above expression. Therefore, we can reformulate the minimization with only inequality constraints on $(\mu_k)_{k=0}^{D-2}$. To be explicit, the minimization becomes
\begin{equation}
    \min_{0\leq \mu_{k<D-1} \leq 1/(k+1)} \,{\max_{k<D-1}\mu_k^{-1}}.
\end{equation}
In Figure \ref{fig:optimal_weights} we illustrate how the minimum is obtained. The gray shaded area is the domain of minimization of $\bm{\mu}$. The equation $\max_{k<D-1} \mu_k^{-1} = a$ is an L-shaped curve (in color red) with its corner on the line defined by $\mu_0 = \mu_{1}=\cdots = \mu_{D-2}$. As the curve moves away from the origin along the line $\mu_0 = \mu_{1}=\cdots = \mu_{D-2}$, the value of $\max_{k<D-1} \mu_k^{-1}$ decreases. It is then clear that the minimum is obtained as the curve touches the slanted boundary of the gray domain, defined by $\mu_0 + 2\mu_1 + \cdots + (D\!-\!1)\mu_{D-2}\!=\!1$ and $\mu_{D-1}=0$. Solving its intersection with the line $\mu_0 = \mu_{1}=\cdots = \mu_{D-2}$, we obtain the optimal weight vector
\begin{equation}
    \bm{w} = \frac{2}{D(D-1)}(D\!-\!1,D\!-\!2,\ldots,1).
\end{equation}

\begin{figure}[t]
    \centering
    \includegraphics[scale=0.3]{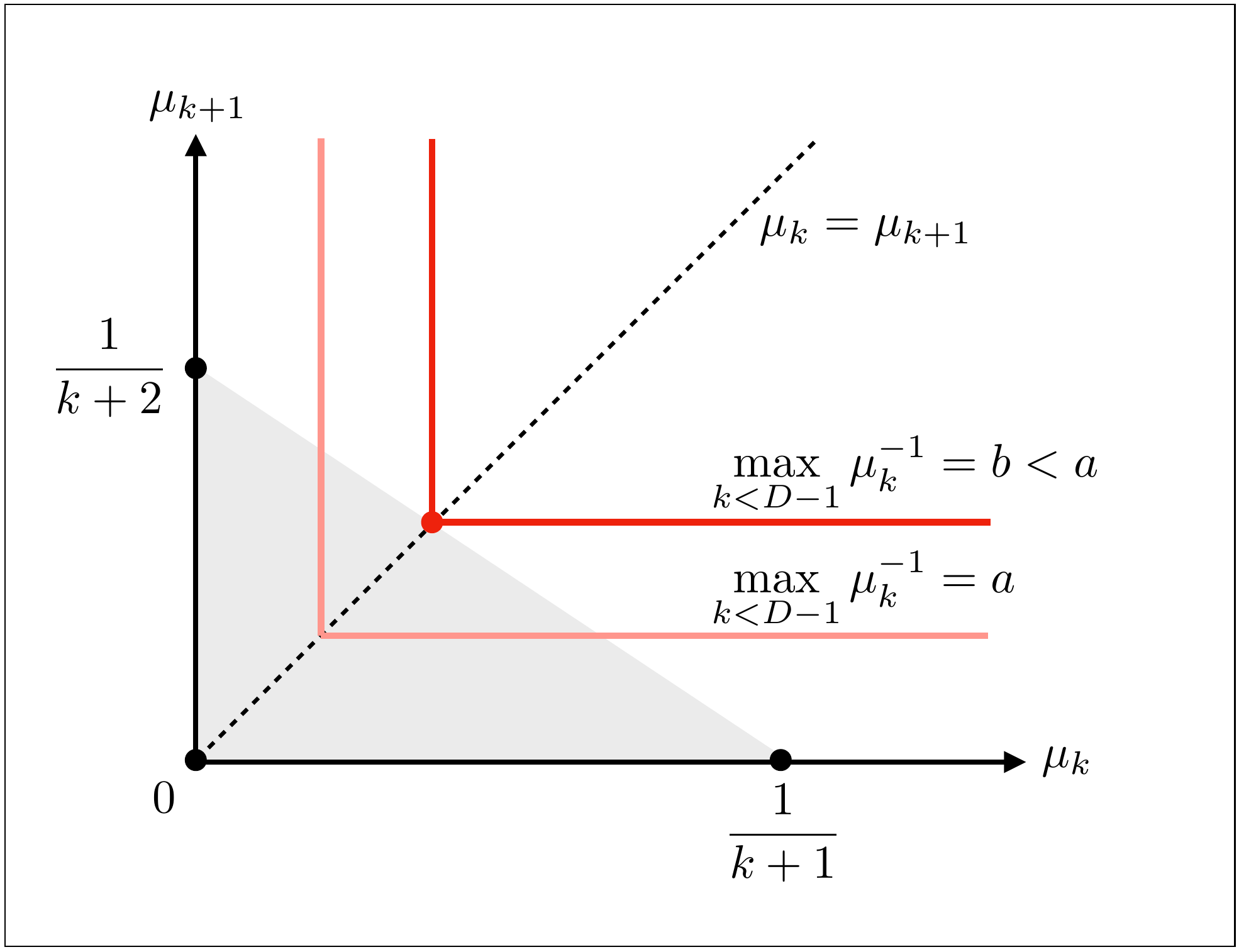}
    \caption{Pictorial illustration of the derivation for optimal weights. See text for more details.}
    \label{fig:optimal_weights}
\end{figure}

When we consider the upper bound for the error of the lowest $K$ eigenstates, even more elements of $\bm{\mu}$ are free and do not enter the minimization
\begin{equation}
    \min_{0\leq\mu_{k<K}\leq1/(k+1)} \, \max_{k<K} \frac{\Theta(K\!-\!k\!-\!1)}{\mu_k}.
\end{equation}
Following the same logic, in the optimal solution the elements $\mu_K,\mu_{K+1},\ldots,\mu_{D-1}$ will be 0. This time, the corner of the curve defined by $\max_{k<K} \mu^{-1}_k\Theta(K\!-\!k\!-\!1)=a$ lies on the line $\mu_0 = \mu_{1}=\cdots = 2\mu_{K-1}$. Combining these condition would then lead to the optimal weights
\begin{equation}
    \bm{w} = \frac{1}{K^2}(2K\!-\!1,2K\!-\!3,\ldots,1).
\end{equation}

As for the optimal weights for the error of the eigenstates, the proof follows similarly. The only difference is that weight gaps $\mu_k$'s are modified by the according energy gaps. For example, for targeting all eigenstates, the optimal weight vector is the minimizer of the following problem
\begin{equation}
    \min_{0\leq \mu_{k<D-1} \leq 1/(k+1)}\,\max_{k<D-1}\mu_k^{-1}(E_{k+1}-E_{k})^{-1}.
\end{equation}
The solution is then given by the intersection of the following conditions
\begin{equation}
    \begin{split}
        &\mu_0 + 2\mu_1 + \cdots (D-1)\mu_{D-2} = 1,
        \\
        &\mu_{D-1} = 0,
        \\
        &\frac{\mu_0}{E_1\!-\!E_0} = \frac{\mu_1}{E_2\!-\!E_1} = \cdots = \frac{\mu_{D-2}}{E_{D-1}\!-\!E_{D-2}}.
    \end{split}
\end{equation}
The first two conditions determine the boundary of the domain, whereas the last condition determines the point on the boundary at which the corner of the curve defined by $\max_{k<D-1}\mu_k^{-1}(E_{k+1}-E_{k})^{-1} = a$ touches. For targeting the lowest $K$ eigenstates, we simply sent all $w_{k\geq K}$ to 0, and replace $E_{D-1}\!-\!E_{D-2}$ with $(E_{D-1}\!-\!E_{D-2})/2$.

\end{document}